\documentclass[a4paper,twocolumn,11pt,accepted=2024-08-06]{quantumarticle}
\pdfoutput=1
\usepackage[utf8]{inputenc}
\usepackage[english]{babel}
\usepackage[T1]{fontenc}
\usepackage{amsmath}
\usepackage[colorlinks=true,citecolor=blue,linkcolor=magenta]{hyperref}
\PassOptionsToPackage{compress}{natbib}
\usepackage[numbers]{natbib}

\usepackage{mathrsfs}
\usepackage{amsfonts}
\usepackage{amssymb}
\usepackage{graphicx}
\usepackage[usenames,dvipsnames]{color}
\usepackage{amsthm}
\usepackage{dsfont}
\usepackage{bm}
\usepackage{tabularx}
\usepackage{tabularray}
\usepackage{comment}

\usepackage{algorithm}
\usepackage[noend]{algpseudocode}

\newtheorem{theorem}{Theorem}

\newtheorem{lemma}{Lemma}

\theoremstyle{definition}
\newtheorem{definition}{Definition}

\theoremstyle{remark}

\renewcommand{\Re}{\mathrm{Re}}
\renewcommand{\Im}{\mathrm{Im}}

\newcommand{\norm}[1]{\Vert #1 \Vert}

\newcommand{\abs}[1]{\vert #1 \vert}
\newcommand{\absLR}[1]{\left\vert #1 \right\vert}

\newcommand{\ket}[1]{\vert{ #1 }\rangle}
\newcommand{\bra}[1]{\langle{ #1 }\vert}
\newcommand{\ketbra}[2]{\vert #1 \rangle \langle #2 \vert}
\newcommand{\braket}[2]{\langle #1 \vert #2 \rangle}
\newcommand{\mean}[1]{\langle #1 \rangle}

\newcommand{\sgn}{\mathrm{sgn}}

\newcommand{\bfH}{\mathbf{H}}
\newcommand{\bfS}{\mathbf{S}}
\newcommand{\bfa}{\mathbf{a}}

\newcommand{\BLUE}[1]{{\color{blue} #1}}

\usepackage{soul}

\renewcommand{\leq}{\leqslant}
\renewcommand{\geq}{\geqslant}

\begin{document}

\title{Measurement-efficient quantum Krylov subspace diagonalisation}

\author{Zongkang Zhang}
\thanks{These two authors contributed equally.}
\affiliation{Graduate School of China Academy of Engineering Physics, Beijing 100193, China}
\orcid{0000-0002-9605-9968}

\author{Anbang Wang}
\thanks{These two authors contributed equally.}
\affiliation{Graduate School of China Academy of Engineering Physics, Beijing 100193, China}
\orcid{0000-0002-3493-5053}

\author{Xiaosi Xu}
\affiliation{Graduate School of China Academy of Engineering Physics, Beijing 100193, China}
\orcid{0000-0002-4894-8322}

\author{Ying Li}
\email{yli@gscaep.ac.cn}
\affiliation{Graduate School of China Academy of Engineering Physics, Beijing 100193, China}
\orcid{0000-0002-1705-2494}

\begin{abstract}
The Krylov subspace methods, being one category of the most important classical numerical methods for linear algebra problems, can be much more powerful when generalised to quantum computing. However, quantum Krylov subspace algorithms are prone to errors due to inevitable statistical fluctuations in quantum measurements. To address this problem, we develop a general theoretical framework to analyse the statistical error and measurement cost. Based on the framework, we propose a quantum algorithm to construct the Hamiltonian-power Krylov subspace that can minimise the measurement cost. In our algorithm, the product of power and Gaussian functions of the Hamiltonian is expressed as an integral of the real-time evolution, such that it can be evaluated on a quantum computer. We compare our algorithm with other established quantum Krylov subspace algorithms in solving two prominent examples. To achieve an error comparable to that of the classical Lanczos algorithm at the same subspace dimension, our algorithm typically requires orders of magnitude fewer measurements than others. Such an improvement can be attributed to the reduced cost of composing projectors onto the ground state. These results show that our algorithm is exceptionally robust to statistical fluctuations and promising for practical applications. 
\end{abstract}

%\begin{abstract}
%Quantum Krylov subspace diagonalisation generalises one of the most important categories of numerical methods, Krylov subspace methods, to quantum computing. However, these methods are prone to errors due to inevitable statistical fluctuation in quantum measurements. 
%To address this problem, we develop a general theoretical framework to analyse the measurement cost in quantum Krylov subspace diagonalisation. According to the framework, we propose a quantum algorithm to construct the Hamiltonian-power Krylov subspace and minimise the measurement cost. In our algorithm, the product of power and Gaussian functions of the Hamiltonian is expressed as an integral of the real-time evolution, such that we can evaluate it on a quantum computer. We compare the measurement cost of our algorithm with other established quantum Krylov algorithms in solving two prominent examples. It is shown that the measurement cost of our algorithm is typically $10^4$ to $10^{12}$ times smaller than other algorithms. Such an improvement in measurement efficiency can be explained by composing the Chebyshev-polynomial projector. 
%These results show that our algorithm is exceptionally robust to statistical fluctuation and promising for practical applications. 
%\end{abstract}

\maketitle

\section{Introduction}

Finding the ground-state energy of a quantum system is of vital importance in many fields of physics~\cite{dagotto1994,wall1995,caurier2005}. The Lanczos algorithm~\cite{lanczos1950,saad1992} is one of the most widely used algorithms to solve this problem. It belongs to the Krylov subspace methods~\cite{bjorck2015}, in which the solution usually converges to the true answer with an increasing subspace dimension. However, such methods are unscalable for many-body systems because of the exponentially-growing Hilbert space dimension~\cite{avella2012}. In quantum computing, there is a family of hybrid quantum-classical algorithms that can be regarded as a quantum generalisation of the classical Lanczos algorithm~\cite{motta2019,yeter2020,parrish2019,stair2020,bespalova2021,cohn2021,klymko2022,cortes2022,epperly2022,shen2022,kyriienko2020,seki2021,kirby2022}. Following Refs.~\cite{stair2020,cortes2022,epperly2022}, we call them quantum Krylov subspace diagonalisation (KSD). These algorithms are scalable with the system size by carrying out the classically-intractable vector and matrix arithmetic on the quantum computer. 
They possess potential advantages as the conventional quantum algorithm to solve the ground-state problem, quantum phase estimation, requires considerable quantum resources~\cite{babbush2018,lee2021}, while the variational quantum eigensolver is limited by the ansatz and classical optimisation bottlenecks~\cite{peruzzo2014, mcclean2018}. 
%They possess potential advantages as \BLUE{other popular quantum algorithms to solve the ground-state problem such as} the quantum phase estimation~\cite{babbush2018,lee2021} requires considerable quantum resources, while the variational quantum eigensolver~\cite{peruzzo2014} is limited by the ansatz and classical optimisation bottlenecks~\cite{mcclean2018}. 
%\BLUE{The quantum KSD algorithms possess potential advantages as the quantum phase estimation~\cite{kitaev1995} requires considerable quantum resources, while the variational quantum eigensolver~\cite{peruzzo2014} is limited by the fixed ansatz and classical optimisation bottlenecks~\cite{mcclean2018}.} 

However, Krylov subspace methods are often confronted with the obstacle that small errors can cause large deviations in the ground-state energy. This issue is rooted in the Krylov subspace that is spanned by an almost linearly dependent basis~\cite{parlett1998}. Contrary to classical computing, in which one can exponentially suppress rounding errors by increasing the number of bits, quantum computing is inherently subject to statistical error. 
Since statistical error decreases slowly with the measurement number $M$ as $\propto 1/\sqrt{M}$, an extremely large $M$ can be required to reach an acceptable error. In this case, although quantum KSD algorithms perform well in principle, the measurement cost has to be assessed and optimised for realistic implementations~\cite{epperly2022,kirby2022}. 
%\BLUE{Since statistical error decreases slowly with the measurement number $M$ as $\propto 1/\sqrt{M}$, quantum KSD may require a large measurement cost. In Ref.~\cite{epperly2022}, the error susceptivity in one of the quantum KSD algorithms is rigorously analysed. However, in general, the measurement cost in the various quantum KSD algorithms remains a question.} 
%\BLUE{In Ref.~\cite{epperly2022}, the error susceptivity in one of the quantum KSD algorithms is rigorously analysed, but the related measurement cost is not considered there. Since statistical error decreases slowly with the measurement number $M$ as $\propto 1/\sqrt{M}$, a large $M$ may be required to reach an acceptable error. In this case, although quantum KSD algorithms perform well in principle, the measurement cost has to be assessed and optimised for realistic implementation. In general, the measurement cost for various quantum KSD algorithms remains a question.}
%However, there are various quantum KSD algorithms, and because of their diversity, in general, the measurement cost remains a question. 

In this work, we present a general and rigorous analysis of the measurement cost in quantum KSD algorithms. Specifically, we obtain an upper bound formula of the measurement number that is applicable to all quantum KSD algorithms. Then, we propose an algorithm to construct the Hamiltonian-power Krylov subspace~\cite{bjorck2015}. In our algorithm, we express the product of Hamiltonian power and a Gaussian function of Hamiltonian as an integral of real-time evolution. In this way, the statistical error decreases exponentially with the power, which makes our algorithm measurement-efficient. We benchmark quantum KSD algorithms by estimating their measurement costs in solving the anti-ferromagnetic Heisenberg model and the Hubbard model. Various lattices of each model are taken in the benchmarking. It is shown that the measurement cost in our algorithm is typically orders of magnitude fewer than other algorithms to reach an error comparable to the classical Lanczos algorithm. We also demonstrate the measurement efficiency of our algorithm by composing a ground-state projector.

\section{Krylov subspace diagonalisation}

First, we introduce the KSD algorithm and some relevant notations. The algorithm starts with a reference state $\ket{\varphi}$ (or a set of reference states \cite{parrish2019,stair2020}; we focus on the single-reference case in this work). Then we generate a set of basis states 
\begin{eqnarray}\label{eq:basis}
\ket{\phi_k} = f_k(H)\ket{\varphi},
\end{eqnarray}
where $H$ is the Hamiltonian, and $f_1,f_2,\ldots,f_d$ are linearly-independent $(d-1)$-degree polynomials (to generate a $d$-dimensional Krylov subspace). For example, it is conventional to take the power (P) function $f_k(H) = H^{k-1}$ in the Lanczos algorithm. These states span a subspace called the Krylov subspace. We compute the ground-state energy by solving the generalised eigenvalue problem 
\begin{eqnarray}\label{eq:eig}
\bfH\bfa = E\bfS\bfa,
\end{eqnarray}
where $\bfH_{k,q} = \bra{\phi_k}H\ket{\phi_q}$ and $\bfS_{k,q} = \braket{\phi_k}{\phi_q}$. Let $E_{min}$ be the minimum eigenvalue of the generalised eigenvalue problem. The error in the ground-state energy is \begin{eqnarray}\label{eq:eK}
\epsilon_K = E_{min} - E_g,
\end{eqnarray}
where $E_g$ is the true ground-state energy. We call $\epsilon_K$ the subspace error. One can also construct generalised Krylov subspaces in which $f_k$ are functions of $H$ other than polynomials; see Table~\ref{table}. Appendix~\ref{app:intro} provides a more detailed introduction to the KSD algorithm.

In the literature, subspaces generated by variational quantum circuits~\cite{parrish2019a,nakanishi2019,huggins2020} and stochastic time evolution~\cite{stair2022} are also proposed.
Quantum KSD techniques can also be used to mitigate errors caused by imperfect gates~\cite{mcclean2017,mcclean2020,yoshioka2022}.
In this work, we focus on Hamiltonian functions because of their similarities to conventional Krylov subspace methods.

\begin{table}[tbhp]
\centering
\begin{tabularx}{\linewidth}{ 
   >{\raggedright\arraybackslash}X 
   >{\centering\arraybackslash}c 
   >{\raggedleft\arraybackslash}X} 
\hline\hline
Abbr. & $f_k(x)$ & Refs. \\ 
\hline
P & $x^{k-1}$ & \cite{seki2021,bespalova2021} \\
CP & $T_{k-1}(x/h_{tot})$ & \cite{kirby2022} \\
GP & $x^{k-1} e^{-\frac{1}{2}x^2\tau^2}$ & This work \\
IP & $x^{-(k-1)}$ & \cite{kyriienko2020} \\
ITE & $e^{-\tau (k-1)x}$ & \cite{motta2019,yeter2020} \\
RTE & $e^{-ix\Delta t\left(k-\frac{d+1}{2}\right)}$ & \cite{parrish2019,stair2020,bespalova2021,cohn2021,klymko2022,cortes2022,epperly2022,shen2022} \\
F & \makebox[1cm][c]{$L^{-1}\sum_{l=1}^L e^{-i\left[x-\Delta E\left(k-1\right)\right]\Delta t\left(l-\frac{L+1}{2}\right)}$} & \cite{cortes2022} \\
\hline\hline
\end{tabularx}
\caption{Operators $f_k(x)$ generating the basis of a (generalised) Krylov subspace. Here $x=H-E_0$, where $E_0$ is a constant. In different algorithms, $f_k$ can be power (P), Chebyshev polynomial (CP), Gaussian-power (GP), inverse power (IP) or exponential [i.e.~imaginary-time evolution (ITE), real-time evolution (RTE) and filter (F)] functions of the Hamiltonian. $T_n$ is the $n$-th Chebyshev polynomial of the first kind, and $E_0 = 0$ in the CP basis. For a Hamiltonian expressed in the form $H = \sum_j h_j \sigma_j$, where $\sigma_j$ are Pauli operators, $h_{tot} = \sum_j \abs{h_j}$ is the 1-norm of coefficients. $\tau$, $\Delta t$ and $\Delta E$ are some real parameters. }
\label{table}
\end{table}

\section{Statistical error and regularisation}

In addition to $\epsilon_K$, the other error source is the statistical error. In quantum KSD algorithms, matrices $\bfH$ and $\bfS$ are obtained by measuring qubits at the end of certain quantum circuits. Measurements yield estimators $\hat{\bfH}$ and $\hat{\bfS}$ of the two matrices, respectively. Errors in $\hat{\bfH}$ and $\hat{\bfS}$ depend on the measurement number. Suppose the budget for each complex matrix entry is $2M$ measurements (the budget for each part is $M$). Variances of matrix entries have upper bounds in the form 
\begin{eqnarray}\label{eq:varH1}
\mathrm{Var}(\hat{\bfH}_{k,q})\leq 2C_{\bfH}^2/M
\end{eqnarray}
and 
\begin{eqnarray}\label{eq:varS1}
\mathrm{Var}(\hat{\bfS}_{k,q})\leq 2C_{\bfS}^2/M,
\end{eqnarray}
where $C_{\bfH}$ and $C_{\bfS}$ are some factors depending on the measurement protocol. For example, suppose each entry of $\bfH$ and $\bfS$ can be expressed in the form $\sum_s q_s\bra{\varphi}U_s\ket{\varphi}$, where $U_s$ are unitary operators, and we measure each term using the Hadamard test~\cite{ekert2002}, then the variance of the entry has the upper bound $2C^2/M$, where $C = \sum_s\abs{q_s}$. 

%\BLUE{Variance upper bounds for real matrices are similar: With $M$ measurements per entry, the upper bounds are in the same form but without the factor of two. For example, for the P basis, $\bfH_{1,1}=\bra{\varphi}H\ket{\varphi}$ can be measured using the method in Ref.~\cite{wecker2015}, where $C_{\bfH}=h_{tot}$. }

We quantify the error in a matrix with the spectral norm. The distributions of $\norm{\hat{\bfH}-\bfH}_2$ and $\norm{\hat{\bfS}-\bfS}_2$ depend on not only the measurement number but also the correlations between matrix entries. We can measure matrix entries independently, then their distributions are uncorrelated. However, certain matrix entries take the same value, and we can measure them collectively. For example, for the P basis, all entries $\bfH_{k-q,q}$ take the same value as $\bfH_{k-1,1}$. We only need to measure one of them on the quantum computer, and then these entries become correlated. In Appendix~\ref{app:norm_distribution}, we analyse the distributions of spectral norms for each basis in Table~\ref{table}. 

Errors in matrices cause an error in the ground-state energy in addition to $\epsilon_K$. It is common that the subspace basis is nearly linearly dependent, which makes the overlap matrix $\bfS$ ill-conditioned. Then, even small errors in matrices may cause a serious error in the ground-state energy. Thresholding is a method to overcome this problem~\cite{motta2019,epperly2022}. In this work, we consider the regularisation method; see Algorithm~\ref{alg:QKSD}. An advantage of the regularisation method is that by taking a proper regularisation parameter $\eta$, the resulting energy $\hat{E}_{min}$ is variational, i.e. $\hat{E}_{min}\geq E_g$ up to a controllable failure probability. 

\begin{table*}[tbhp]
\centering
\begin{tblr}{|c|c|c|c|c|c|}
\hline
Protocol & $\eta$  & $\alpha$ & $\beta$ & Basis \\ \hline
IM (Chebyshev) & $\frac{2d^2}{\sqrt{M\kappa}}$ & $\frac{256}{\kappa}$ & $d^6$ & \SetCell[r=2]{c}All \\ \cline{1-4}
IM (Hoeffding) & $\sqrt{\frac{2d^2}{M}\ln{\frac{8d^2}{\kappa}}}$ & $128\ln\frac{1}{\kappa}$ & $d^4$ &   \\ \hline
CM (Real Hankel) & \SetCell[r=2]{c}$\sqrt{\frac{2d}{M}\ln{\frac{4d}{\kappa}}}$  & $64\ln\frac{1}{\kappa}$ & $d(2d-1)$ &  P, GP, IP, ITE \\ \cline{1-1} \cline{3-5} 
CM (Real symmetric)	&  & $32\ln\frac{1}{\kappa}$ & $d^2(d+1)$ & CP, F \\ \hline
CM (Complex Toeplitz)	& $\sqrt{\frac{2(2d-1)}{M}\ln{\frac{4d}{\kappa}}}$ & $64\ln\frac{1}{\kappa}$ & $(2d-1)^2$ &   RTE \\ \hline
\end{tblr}
\caption{The regularisation parameter $\eta$, and measurement cost factors $\alpha(\kappa)$ and $\beta(d)$. These quantities depend on the protocol for measuring matrix entries. Two types of protocols are considered: independent measurement (IM) and collective measurement (CM). Using the Chebyshev and Hoeffding inequalities, we obtain two different results (upper bounds) for the IM protocol. These factors are based on the protocols introduced in Appendices \ref{app:norm_distribution} and \ref{app:optimisation}. We remark that for CP, leveraging the characteristic of Chebyshev polynomials, the scaling with $d$ can be reduced to $O(d^2)$ \cite{kirby2022}.}
\label{tab:eta}
\end{table*}

We propose to choose the regularisation parameter $\eta$ according to the distributions of $\norm{\hat{\bfH}-\bfH}_2$ and $\norm{\hat{\bfS}-\bfS}_2$. Let $\kappa$ be the permissible failure probability. We take $\eta$ such that 
\begin{eqnarray}
&& \Pr\left(\norm{ \hat \bfH - \bfH}_2 \leq C_{\bfH}\eta \ \rm{and}\ \norm{ \hat \bfS - \bfS}_2 \leq C_{\bfS}\eta \right) \notag \\
&\geq & 1-\kappa.
\label{eq:condition}
\end{eqnarray}
In Table~\ref{tab:eta}, we list $\eta$ satisfying the above condition for each basis. The value of $\eta$ depends on the failure probability $\kappa$ and the measurement number $M$. Taking $\eta$ in this way, an upper bound of the energy error is given in the following lemma. 

\begin{lemma}\label{lem}
Let 
\begin{eqnarray}
E'(\eta,\bfa) = \frac{\bfa^\dag (\bfH+2C_{\bfH}\eta) \bfa}{\bfa^\dag (\bfS+2C_{\bfS}\eta) \bfa}.
\label{eq:Ep}
\end{eqnarray}
Under conditions $E_g<0$ and $\min_{\bfa\neq \bf{0}} E'(\eta,\bfa)<0$, 
the following inequality holds, 
\begin{eqnarray}\label{eq:hatE_bound}
\Pr\left(E_g \leq \hat{E}_{min} \leq \min_{\bfa\neq \bf{0}} E'(\eta,\bfa)\right) \geq 1-\kappa. 
\end{eqnarray}
\end{lemma}

See Appendix~\ref{app:proof} for the proof. Notice that we can always subtract a sufficiently large positive constant from the Hamiltonian to satisfy the conditions. The upper bound of the energy error in the regularisation method is $\min_{\bfa\neq \bf{0}} E'(\eta,\bfa) - E_g$, up to the failure probability.

\begin{figure}
\begin{minipage}{\linewidth}
\begin{algorithm}[H]
{\small
\begin{algorithmic}[1]
\caption{{\small Regularised quantum Krylov-subspace diagonalisation algorithm.}}
\label{alg:QKSD}
\Statex
\State Input $\hat{\bfH}$, $\hat{\bfS}$, $C_{\bfH}$, $C_{\bfS}$ and $\eta>0$. 
\State Solve the generalised eigenvalue problem 
$$(\hat{\bfH}+C_{\bfH}\eta)\bfa = E(\hat{\bfS}+C_{\bfS}\eta)\bfa$$
\State Output the minimum eigenvalue $\hat{E}_{min}$. 
\end{algorithmic}
}
\end{algorithm}
\end{minipage}
\end{figure}

\section{Analysis of the measurement cost}

In this work, we rigorously analyse the number of measurements. Let $\epsilon$ be the target error in the ground-state energy. We will give a measurement number sufficient (i.e. an upper bound of the measurement number necessary) for achieving the target. For the regularisation method, we already have an upper bound of the error $\min_{\bfa\neq \bf{0}} E'(\eta,\bfa) - E_g$. Therefore, we can derive such a measurement number by solving the equation 
\begin{eqnarray}
\min_{\bfa\neq \bf{0}} E'(\eta,\bfa) - E_g = \epsilon,
\label{eq:equation}
\end{eqnarray}
in which $\eta$ is the unknown. With the solution, we can work out the corresponding measurement number $M$ according to Table~\ref{tab:eta}. If matrix entries are measured independently, the total measurement number is $M_{tot} = 4d^2M$; we have the following theorem.

\begin{theorem}
Suppose $\epsilon>\epsilon_K$ and $E_g+\epsilon<0$. The total measurement number 
\begin{eqnarray}
M_{tot} = \frac{\alpha(\kappa)\beta(d)}{16\eta^2}
\label{eq:Mtot}
\end{eqnarray}
is sufficient for achieving the permissible error $\epsilon$ and failure probability $\kappa$, i.e.~$\hat{E}_{min}\in [E_g,E_g+\epsilon]$ with a probability of at least $1-\kappa$. Here, $\alpha(\kappa) = 256/\kappa$, $\beta(d) = d^6$, and $\eta$ is the solution to Eq. (\ref{eq:equation}). 
\label{the}
\end{theorem}

See Appendix~\ref{app:proof} for the proof. Similar to Lemma \ref{lem}, we can always satisfy the condition $E_g+\epsilon<0$ by subtracting a sufficiently large positive constant from the Hamiltonian. The above theorem also holds (up to some approximations) for collective measurements if we replace functions $\alpha(\kappa)$ and $\beta(d)$ according to Table~\ref{tab:eta}.

\subsection{Factoring the measurement cost}

To compare different bases, we would like to divide the measurement number into three factors. If matrix entries are measured independently, the first two factors are the same for different bases of Krylov subspaces, and the third factor depends on the basis. We find that the third factor is related to the ill-conditioned problem, and we can drastically reduce it with our measurement-efficient algorithm.

We can always write the total measurement number Eq.~(\ref{eq:Mtot}) as a product of three factors, 
\begin{eqnarray}
M_{tot} = \frac{\alpha(\kappa)\norm{H}_2^2}{p_g^2\epsilon^2}
\times \beta(d) \times \gamma.
\label{eq:Mtot_factoring}
\end{eqnarray}
The first factor is the cost of measuring the energy. Roughly speaking, it is the cost in the ideal case that one can prepare the ground state by an ideal projection; see Appendix~\ref{app:projection}. The spectral norm $\norm{H}_2$ characterises the range of the energy. Assume that the statistical error and $\epsilon$ are comparable, $M_{tot}\propto \norm{H}_2^2/\epsilon^2$ according to the standard scaling of the statistical error. $\alpha(\kappa)$ is the overhead for achieving the permissible failure probability $\kappa$. The success of KSD algorithms depends on a finite overlap between the reference state and the true ground state $\ket{\psi_g}$~\cite{bjorck2015,epperly2022,kirby2022}. This requirement is reflected in $M_{tot}\propto 1/p_g^2$, where $p_g = \abs{\braket{\psi_g}{\varphi}}^2$. The second factor $\beta(d)$ is the overhead due to measuring $d\times d$ matrices. There are more sources of statistical errors (matrix entries) influencing the eventual result when $d$ is larger. The remaining factor 
\begin{eqnarray}
\gamma = \frac{p_g^2\epsilon^2}{16\norm{H}_2^2\eta^2}
\label{eq:gamma}
\end{eqnarray}
is related to the spectrum of $\bfS$ (i.e.~the basis), as we will show in Section~\ref{sec:projector}. Notice that this expression of $\gamma$ holds for all the bases. In the numerical result, we find that the typical value of $\gamma$ for certain bases can be as large as $10^{12}$ to achieve certain $\epsilon$. We propose the Gaussian-power (GP) basis (see Table \ref{table}), and it can reduce $\gamma$ to about $4$.

\section{Measurement-efficient algorithm}

As we have shown, the measurement overhead $\gamma$ depends on the overlap matrix $\bfS$, i.e.~how we choose the basis of Krylov subspace. The main advantage of our algorithm is that we utilise a basis resulting in a small $\gamma$. 

To generate a standard Krylov subspace, we choose operators 
\begin{eqnarray}
f_k = (H-E_0)^{k-1}e^{-\frac{1}{2}(H-E_0)^2\tau^2},
\label{eq:fk}
\end{eqnarray}
where $E_0$ is a constant up to choice. We call it the GP basis. The corresponding subspace is a conventional Hamiltonian-power Krylov subspace, but the reference state $\ket{\varphi}$ has been {\it effectively} replaced by $e^{-\frac{1}{2}(H-E_0)^2\tau^2}\ket{\varphi}$. 
%We will show that the Gaussian-power basis is measurement-efficient: $\gamma$ can be smaller by orders of magnitude to achieve the same permissible error than other bases in Table~\ref{table}. 

The spectral norm of $f_k$ operators for the GP basis [defined in Eq.~(\ref{eq:fk})] has the upper bound 
\begin{eqnarray}\label{eq:norm_bound}
\norm{f_k}_2\leq (\frac{k-1}{e\tau^2})^{\frac{k-1}{2}},     
\end{eqnarray}
where $e$ is Euler's constant. This upper bound is proven in Lemma~\ref{lem:norm} in Appendix~\ref{app:GPbasis}. Under the condition $e\tau^2>d-1$, the norm upper bound $(\frac{k-1}{e\tau^2})^{\frac{k-1}{2}}$ decreasing exponentially with $k$. This property is essential for the measurement efficiency of our algorithm.

We propose to realise the Gaussian-power basis by expressing the operator of interest as a linear combination of unitaries (LCU). The same approach has been taken to realise other bases like power~\cite{seki2021,bespalova2021}, inverse power~\cite{kyriienko2020}, imaginary-time evolution~\cite{huo2021}, filter~\cite{cortes2022} and Gaussian function of the Hamiltonian~\cite{zeng2021}.

Suppose we want to measure the quantity $\bra{\varphi}A\ket{\varphi}$. Here $A = f_k^\dag Hf_q$ and $A = f_k^\dag f_q$ for $\bfH_{k,q}$ and $\bfS_{k,q}$, respectively. First, we work out an expression in the from $A = \sum_s q_s U_s$, where $U_s$ are unitary operators, and $q_s$ are complex coefficients. Then, there are two ways to measure $\bra{\varphi}A\ket{\varphi}$. When the expression has finite terms, we can apply $A$ on the state $\ket{\varphi}$ by using a set of ancilla qubits; In the literature, LCU usually refers to this approach~\cite{childs2012}. Alternatively, we can evaluate the summation $\bra{\varphi}A\ket{\varphi}$ using the Monte Carlo method~\cite{faehrmann2022} by averaging Hadamard-test~\cite{ekert2002} estimates of randomly selected terms $\bra{\varphi}U_s\ket{\varphi}$. The second way works for an expression with even infinite terms and requires only one or zero ancilla qubits~\cite{o2021,lu2021}. We focus on the Monte Carlo method in this work. 

Now we give our LCU expression. Suppose we already express $H$ as a linear combination of Pauli operators, it is straightforward to work out the LCU expression of $A$ given the expression of $f_k$. Therefore, we only give the expression of $f_k$. Utilising the Fourier transformation of a modified Hermite function (which is proved in Appendix~\ref{app:GPbasis}), we can express $f_k$ in Eq.~(\ref{eq:fk}) as 
\begin{eqnarray}
f_k &=& \frac{i^{k-1}}{2^{\frac{k-1}{2}}\tau^{k-1}}
\int_{-\infty}^{+\infty}dt H_{k-1}\left(\frac{t}{\sqrt{2}\tau}\right)g_\tau(t) \notag\\
&&\times e^{-ixt},
\label{eq:LCU}
\end{eqnarray}
where $H_n(u)$ denotes Hermite polynomials~\cite{arfken2013}, $g_\tau(t) = \frac{1}{\tau\sqrt{2\pi}}e^{-\frac{t^2}{2\tau^2}}$ is the normalised Gaussian function, $x = H-E_0$, and $e^{-ixt}$ is the real-time evolution (RTE) operator. Notice that in Eq.~(\ref{eq:LCU}) $f_k$ is expressed as a linear combination of RTE operators $e^{-iHt}$, and the summation in conventional LCU is replaced with an integral over $t$. See Appendix~\ref{app:algorithm} for how to evaluate this integral using the Monte Carlo method. 

\subsection{Real-time evolution}
In order to implement the LCU expression Eq.~(\ref{eq:LCU}), we need to implement the RTE $e^{-iHt}\ket{\varphi}$.
There are various protocols for implementing RTE, including Trotterisation~\cite{lloyd1996,berry2007,wiebe2010}, LCU~\cite{childs2012,berry2015,meister2022,faehrmann2022} and qDRIFT~\cite{campbell2019}, etc. In most of the protocols, RTE is inexact but approaches the exact one when the circuit depth increases. Therefore, all these protocols work for our algorithm as long as the circuit is sufficiently deep. In this work, we specifically consider an LCU-type protocol, the zeroth-order leading-order-rotation formula~\cite{yang2021}. 

In the leading-order-rotation protocol, RTE is {\it exact} even when the circuit depth is finite. In this protocol, we express RTE in the LCU form 
\begin{eqnarray}\label{eq:LOR}
e^{-iHt} = \left[\sum_r v_r(t/N) V_r(t/N)\right]^N,     
\end{eqnarray}
where $V_r(t) = e^{-i\phi(t)\sigma},\sigma$ is a rotation or Pauli operator, $v_r(t)$ are complex coefficients, and $N$ is the time step number. This exact expression of RTE is worked out in the spirit of Taylor expansion and contains infinite terms. Notice that $N$ determines the circuit depth. For details of the zeroth-order leading-order-rotation formula, see Ref.~\cite{yang2021} or Appendix~\ref{app:algorithm} in this paper, where Eq.~(\ref{eq:LOR2}) gives the specific form of Eq.~(\ref{eq:LOR}). Substituting Eq.~(\ref{eq:LOR}) into Eq.~(\ref{eq:LCU}), we obtain the eventual LCU expression of $f_k$, a concatenation of two LCU expressions, i.e.
\begin{eqnarray}\label{eq:LCU2}
f_k &=& \frac{i^{k-1}}{2^{\frac{k-1}{2}}\tau^{k-1}}
\int_{-\infty}^{+\infty}dt H_{k-1}\left(\frac{t}{\sqrt{2}\tau}\right)g_\tau(t) e^{i E_0 t} \notag\\
&&\times \left[\sum_r v_r(t/N) V_r(t/N)\right]^N.
\end{eqnarray}
With the above expression, we have written $f_k$ as the linear combination of $V_r(t/N)$. Therefore, we can use the Monte Carlo method to realise $f_k$: We sample $t$ and $r$ according to the coefficients, then we implement the corresponding $V_r(t/N)$ on the quantum computer.

\subsection{Bias and variance}

Let $\hat{A}$ be the estimator of $\bra{\varphi}A\ket{\varphi}$. There are two types of errors in $\hat{A}$, bias and variance. Because RTE is exact in the leading-order-rotation protocol, the estimator $\hat{A}$ is unbiased. Therefore, we can focus on the variance from now on. 

If we use the Monte Carlo method and the one-ancilla Hadamard test to evaluate the LCU expression of $A$, 
the variance has the upper bound 
\begin{eqnarray}\label{eq:variance_bound}
\mathrm{Var}(\hat{A}) \leq \frac{2C_A^2}{M},
\end{eqnarray}
where $2M$ is the measurement number, and $C_A = \sum_s \abs{q_s}$ is the 1-norm of coefficients in $A=\sum_s q_s U_s$. Here, $A$ is a summation of unitary operators. The generalisation to integral is straightforward. We call $C_A$ the cost of the LCU expression. We remark that the variance in this form is universal and valid for all algorithms with some proper factor $C_A$. 

The cost of an LCU expression is related to the spectral norm. Notice $\norm{U_s}_2 = 1$, we immediately conclude that $C_A\geq \norm{A}_2$. Therefore, the best we can achieve is a carefully designed LCU expression that $C_A$ is as close as possible to $\norm{A}_2$. Only when $\norm{f_k}_2$ decreases exponentially with $k$, it is possible to measure $\bra{\varphi}A\ket{\varphi}$ with an error decreasing exponentially with $k$.

In our algorithm, we find that the cost has the form 
\begin{eqnarray}
C_A = h_{tot}c_kc_q\text{ and }c_kc_q
\end{eqnarray}
for $\bfH_{k,q}$ and $\bfS_{k,q}$, respectively. $c_k$ is the cost due to $f_k$, and 
\begin{eqnarray}
c_k &=& \frac{1}{2^{\frac{k-1}{2}}\tau^{k-1}}
\int_{-\infty}^{+\infty}dt \absLR{H_{k-1}\left(\frac{t}{\sqrt{2}\tau}\right)}g_\tau(t) \notag \\
&&\times \left[\sum_r \absLR{v_r\left(\frac{t}{N}\right)}\right]^N
\leq 2 \left(\frac{k-1}{e\tau^2}\right)^{\frac{k-1}{2}}.
\label{eq:ck}
\end{eqnarray}
This upper bound holds when $N \geq 4eh_{tot}^2\tau^2$, which is a requirement on the circuit depth $N$. The proof is in Appendix~\ref{app:algorithm}. With the upper bound, we can find that we already approach the lower bound of the variance given by the spectral norm (The difference is a factor of $2$). As a result, the variance of $\hat{A}$ decreases exponentially with $k$ and $q$.

%which is effectively a limitation on the value of $\tau$ that we can take in the algorithm, i.e. the largest value of $\tau$ depends on the circuit depth.

Detailed pseudocodes and analysis are given in Appendix~\ref{app:algorithm}. There are two ways to apply Theorem~\ref{the}: First, we can take $C_{\bfH} = h_{tot}C_{\bfS}$ and $C_{\bfS} = \max\{c_kc_q\}$; Second, we can rescale the basis by taking $f'_k = f_k/c_k$. Matrix entries $\hat{\bfH}'_{k,q}$ and $\hat{\bfS}'_{k,q}$ of the rescaled basis $\{f'_k\ket{\varphi}\}$ have variance upper bounds $2h_{tot}^2/M$ and $2/M$, respectively. Therefore, $C_{\bfH} = h_{tot}$ and $C_{\bfS} = 1$ for the rescaled basis. We take the rescaled basis in the numerical study. 

\section{Measurement overhead benchmarking}\label{sec:benchmarking}

With Theorem~\ref{the}, we can benchmark the measurement cost in quantum KSD algorithms listed in Table~\ref{table}. We focus on the overhead factor $\gamma$, while ignoring the slight differences in the other two factors in different algorithms (See Table~\ref{tab:eta}). Given a target error $\epsilon$, first we work out the solution $\eta$ to Eq.~(\ref{eq:equation}), then $\gamma$ is calculated according to Eq.~(\ref{eq:gamma}).

Two models of strongly correlated systems, namely, the anti-ferromagnetic Heisenberg model and Hubbard model~\cite{avella2012,anderson1987,lee2006,seki2020}, are used in benchmarking. For each model, the lattice is taken from two categories: regular lattices, including chain and ladder, and randomly generated graphs. We fix the lattice size such that the Hamiltonian can be encoded into ten qubits. For each Hamiltonian specified by the model and lattice, we evaluate all quantum KSD algorithms with the same subspace dimension $d$ taken from $2,3,\ldots,30$. In total, $233$ instances of $(model,lattice,d)$ are generated to benchmark quantum KSD algorithms. The details of the numerical calculations including parameters taken in different algorithms are given in Appendix~\ref{app:details}.

Fig.~\ref{fig:error_vs_overhead} illustrates the result of the instance $(\mathrm{Heisenberg},\mathrm{chain},d=5)$, for example. We can find that the error decreases with the cost overhead. When the target error is relatively high, algorithms have similar performance, e.g. when $\epsilon \approx 0.01$, the cost $\gamma \approx 10$ is sufficient in most algorithms. However, in comparison with others, the error in the GP algorithm decreases much faster.
When $d$ is sufficiently large, $M$ required to achieve the error $\epsilon$ always follows a scaling of $1/\epsilon^2$. However, the size of $d$ at which different algorithms converge to this behavior varies; see Fig.~\ref{fig:scaling}. The GP algorithm converges with a smaller $d$ compared to other algorithms. In this instance, the GP algorithm has already converged when $d=5$, whereas the other algorithms have not yet converged, resulting in poorer performance.

\begin{figure}[htbp]
\centering
\includegraphics[width=\linewidth]{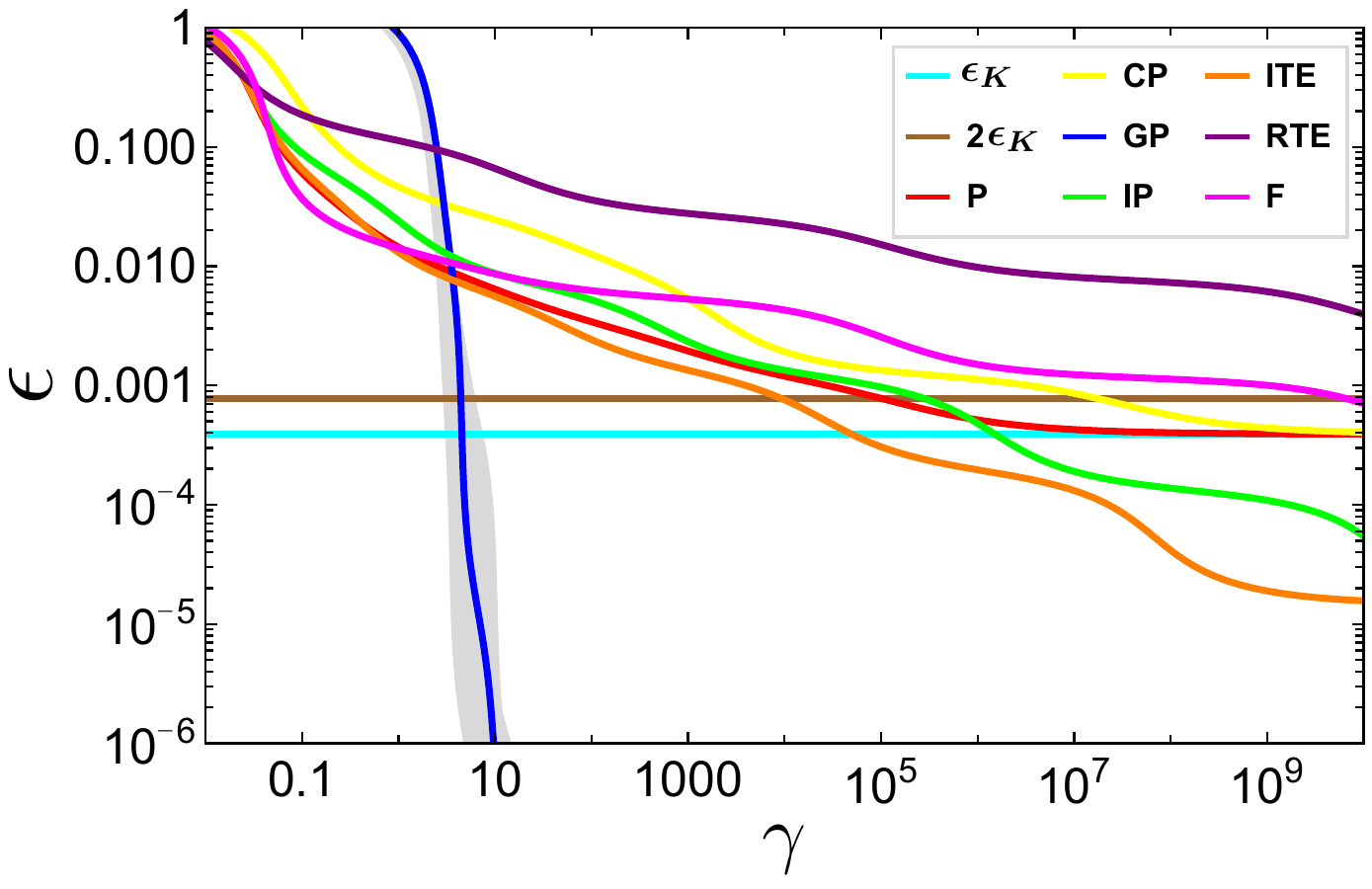}
\caption{Comparison between quantum KSD algorithms listed in Table~\ref{table}. Here, we take the instance $(\mathrm{Heisenberg},\mathrm{chain},d=5)$ as an example. 
$\epsilon_K$ is the subspace error of the P algorithm. The grey area illustrates the range of $\gamma$ when we take $E_0\in [E_g-0.1,E_g+0.1]$ in the GP algorithm. Notice that $\norm{H}_2 = 1$. The blue curve represents the result of the GP algorithm with $E_0=E_g$. Similar results are obtained when other measures are used for benchmarking, including the measurement number $M$ computed with the optimal $\alpha$ and $\beta$ factors, the necessary measurement cost (instead of its upper bound) estimated numerically and the necessary cost when using the thresholding method; see Appendix~\ref{app:practice}.}
\label{fig:error_vs_overhead}
\end{figure}

To confirm this advantage of the GP algorithm, we implement the simulation for a number of instances. We choose the power algorithm as the standard for comparison. We remark that the power algorithm and the classical Lanczos algorithm yield the same result when the implementation is error-free (i.e.~without statistical error and rounding error), and the eventual error in this ideal case is the subspace error $\epsilon_K$. Given the model, lattice and subspace dimension, we compute $\epsilon_K$ of the power algorithm. Then, we take the permissible error $\epsilon = 2\epsilon_K$ and compute the overhead factor $\gamma$ for each algorithm. Then, the empirical distribution of $\gamma$ is shown in Fig.~\ref{fig:distribution}.

\begin{figure}[htbp]
\centering
\includegraphics[width=\linewidth]{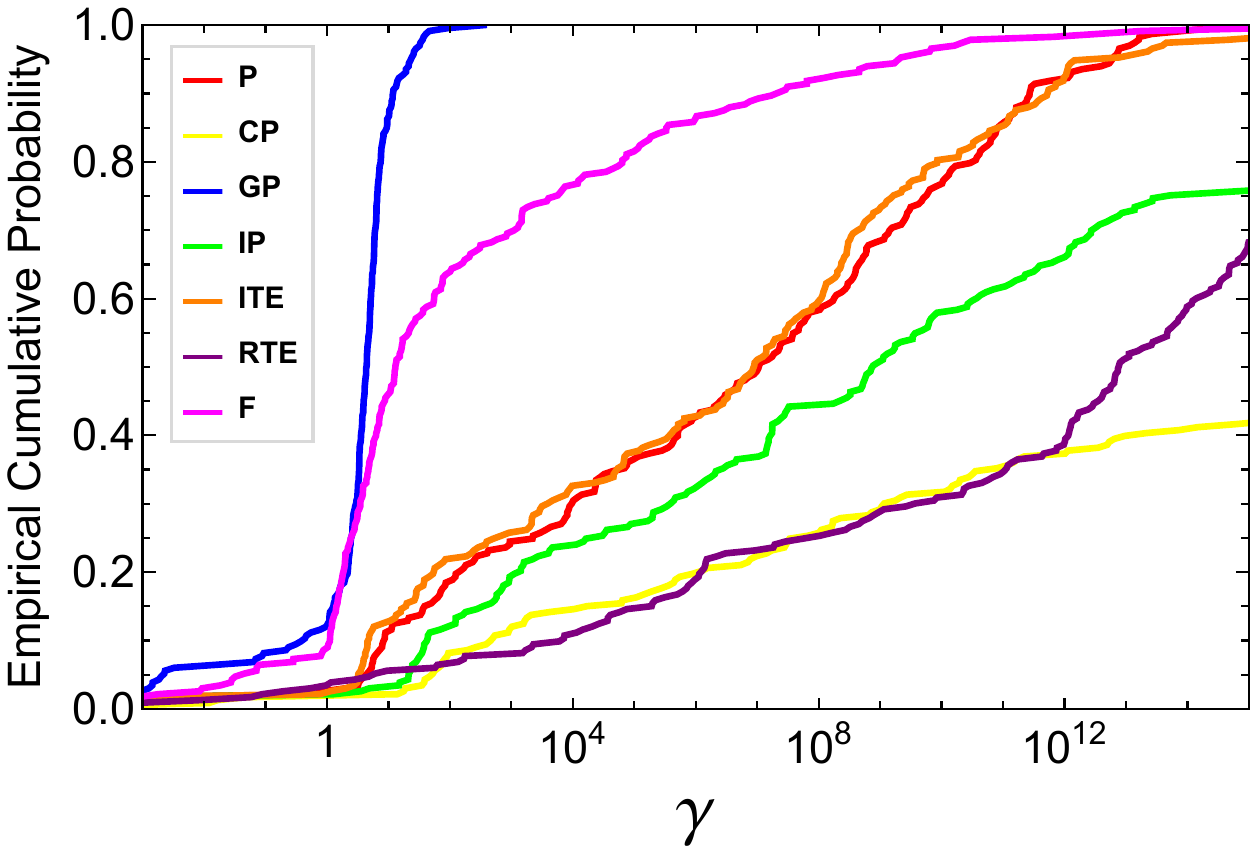}
\caption{Empirical distribution of the measurement overhead $\gamma$ for algorithms listed in Table~\ref{table}. In the Gaussian-power algorithm, we take a random $E_0$ in the interval $[E_g-0.1\norm{H}_2,E_g+0.1\norm{H}_2]$, i.e.~we assume that we have a preliminary estimation of the ground-state energy with an uncertainty as large as $10\%$ of the entire spectrum. }
\label{fig:distribution}
\end{figure}

From the distribution, we can find that our Gaussian-power algorithm has the smallest measurement overhead, and $\gamma$ is smaller than $10^2$ for almost all instances. The filter (F) algorithm is the most well-performed one among others. Part of the reason is that we have taken the optimal parameter found by grid search, of which the details are discussed in Appendix~\ref{app:details}. The median value of $\gamma$ is about $4$ for the GP algorithm, $13$ for the F algorithm and as large as $10^{12}$ for some other algorithms. Comparing the GP and F algorithms, although the median values are similar, the cost overhead is larger than $10^4$ in more than $20\%$ instances in the F algorithm, which never happens in the GP algorithm. 

The performance of the GP algorithm depends on the choice of parameters: $\tau$, $E_0$ and $\eta$. The regularisation parameter $\eta$ is taken according to Table~\ref{tab:eta}. To choose the value of $E_0$, we need prior knowledge about the true ground-state energy $E_g$, which can be obtained from certain classical algorithms computing the ground state or the quantum KSD algorithm taking a smaller subspace dimension. Although we prefer an $E_0$ close to $E_g$, the numerical result suggests that a relatively large difference $\abs{E_0-E_g}$ ($10\%$ of the entire spectrum) is tolerable. The Gaussian factor in the GP basis acts as a filter on the spectrum. Therefore, if $\abs{E_0-E_g}$ is finite and we take an over large $\tau$, we may have a vanishing overlap between the effective reference state $e^{-\frac{1}{2}(H-E_0)^2\tau^2}\ket{\varphi}$ and the ground state. To avoid this problem, we need to assume an upper bound $\abs{E_0-E_g} \leq \epsilon_0$, then we suggest taking $\tau$ in the range $\sqrt{\frac{d-1}{e}} < \tau \leq \frac{\sqrt{d-1}}{\epsilon_0}$, see Appendix~\ref{app:tau}. Although we take the same $\tau$ for all $f_k$ in each instance in the benchmarking, one can flexibly choose different $\tau$ in practice.

%\BLUE{The performance of the GP algorithm depends on the choice of parameters. One needs to shift the Hamiltonian such that $E_g+\epsilon<0$. When the Hamiltonian is expressed using Pauli operators, this condition can be satisfied by removing the constant term. The regularisation parameter $\eta$ is taken according to Table~\ref{tab:eta} and CM protocol is recommended, in which $\eta$ is much smaller. We hope $E_0$ is near the true ground-state energy $E_g$, which requires the prior knowledge. The value of $\tau$ has an upper bound given the circuit depth. To have a good performance, we suggest $\tau$ in the range $\left(\sqrt{\frac{d-1}{e}},\frac{\sqrt{d-1}}{\abs{E_g-E_0}}\right)$. The details of parameter choices is in Appendix~\ref{app:practice}.}

\section{Composing a projector}\label{sec:projector}
We show that $\gamma$ is related to the spectrum of $\bfS$. Let's consider a small $\eta$ and the Taylor expansion of $E'$. The minimum value of the zeroth-order term $\frac{\bfa^\dag \bfH \bfa}{\bfa^\dag \bfS \bfa}$ is $E_{min}$ according to the Rayleigh quotient theorem~\cite{horn2012}. Take this minimum value, we have $E' \simeq E_{min} + s\eta$, where $s\propto\frac{\bfa^\dag \bfa}{\bfa^\dag \bfS \bfa}$ is the first derivative. The solution to $E' = E_g+\epsilon$ is $\eta \simeq (E_g+\epsilon-E_{min})/s$, then 
\begin{eqnarray}\label{eq:gamma_S}
\gamma \simeq u^2\left(\frac{p_g\bfa^\dag \bfa}{\bfa^\dag \bfS \bfa}\right)^2, 
\end{eqnarray}
where $u\leq \frac{(C_{\bfH}+\norm{H}_2C_{\bfS})\epsilon}{2\norm{H}_2(\epsilon-\epsilon_K)}\approx C_{\bfS}$ under assumptions $C_{\bfH}\approx\norm{H}_2C_{\bfS}$ and $\epsilon-\epsilon_K\approx\epsilon$. Eq.~(\ref{eq:gamma_S}) is derived in Appendix~\ref{app:gamma}. Therefore, when $\bfS$ is close to singular, the overhead can be large. 

We can explain the measurement efficiency of our algorithm by composing a projector. In our algorithm with the rescaled basis, the solution is a state in the form $\ket{\Psi(\bfa)} = \sum_{k=1}^d a_kc_k^{-1}f_k\ket{\varphi}$. Ideally, the linear combination of $f_k$ realises a projector onto the ground state, i.e.~$\sum_{k=1}^d a_kc_k^{-1}f_k = \ketbra{\psi_g}{\psi_g}$. Then, $\ket{\Psi(\bfa)} = \sqrt{p_g}\ket{\psi_g}$ up to a phase and $\bfa^\dag \bfS \bfa = \braket{\Psi(\bfa)}{\Psi(\bfa)} = p_g$. Then $\gamma\approx \left(\frac{p_g\bfa^\dag \bfa}{\bfa^\dag \bfS \bfa}\right)^2 = \left(\bfa^\dag \bfa\right)^2$ (notice that $C_\bfS = 1$), we can find that $\bfa^\dag \bfa$ determines $\gamma$. When $c_k$ is smaller, $\abs{a_k}$ required for composing the projector is smaller. In our algorithm, $c_k$ decreases exponentially with $k$, and the speed of decreasing is controllable via the parameter $\tau$. In this way, our algorithm results in a small $\gamma$. 

To be concrete, let's consider a classic way of composing a projector from Hamiltonian powers using Chebyshev polynomials (CPs) $T_n$, which has been used for proving the convergence of the Lanczos algorithm~\cite{bjorck2008}, RTE and CP quantum KSD algorithms~\cite{epperly2022,kirby2022}. Without loss of generality, we suppose $\norm{H}_2 = 1$. An approximate projector in the form of Chebyshev polynomials is 
\begin{eqnarray}\label{eq:projector}
\frac{T_n(Z)}{T_n(z_1)} = \ketbra{\psi_g}{\psi_g} + \Omega,
\end{eqnarray}
where $Z = 1-(H-E_g-\Delta)$, $\Delta$ is the energy gap between the ground state and the first excited state, $z_1 = 1+\Delta$, and $\Omega$ is an operator with the upper bound $\norm{\Omega}_2\leq 2/(z_1^n+z_1^{-n})$. Notice that the error $\Omega$ depends on $n$, and its upper bound decreases exponentially with $n$ if $\Delta$ is finite. For simplicity, we focus on the case that $E_0 = E_g$ in the Hamiltonian power. Then, in the expansion $T_n(Z) = \sum_{l=0}^n b_l (H-E_g)^l$, coefficients have upper bounds $\abs{b_l}\leq n^l T_n(z_1)$ increasing exponentially with $l$. See Appendix~\ref{app:CPP}. In the LCU and Monte Carlo approach, large coefficients lead to large variance. Therefore, it is difficult to realise this projector because of the exponentially increasing coefficients. 

In our algorithm, the exponentially decreasing $c_k$ can cancel out the large $b_l$. We can compose the CP projector with an additional Gaussian factor. Taking $d = n+1$ and $a_k = c_kb_{k-1}/T_n(z_1)$, we have $\sum_{k=1}^d a_kc_k^{-1}f_k = \frac{T_n(Z)}{T_n(z_1)}e^{-\frac{1}{2}(H-E_g)^2\tau^2}$. Because the Gaussian factor is also a projector onto the ground state, the overall operator is a projector with a smaller error than the CP projector. The corresponding overhead factor is 
\begin{eqnarray}
\gamma \lesssim 4\frac{1}{1-n^3/(e\tau^2)}.
\end{eqnarray}
When $\tau$ is sufficiently large, $\gamma \lesssim 4$. This example shows that the Gaussian-power basis can compose a projector with a small measurement cost. 

By composing a projector with an explicit expression, we can obtain the worst-case performance of the algorithm. Similar analyses are reported in Refs.~\cite{epperly2022} and \cite{kirby2022} for RTE and CP algorithms, respectively. Because the projector is approximate, it results in a finite error in the ground-state energy $\epsilon_P$. By solving the generalised eigenvalue problem, we can usually work out a better solution, i.e. the corresponding energy error is lower than $\epsilon_P$. Our result about the measurement cost is applicable to a target error that is not limited by $\epsilon_P$. Theorem~\ref{the} holds even when $\epsilon<\epsilon_P$, which is beyond the explicit projector analysis.

The method of constructing polynomial projection operators mentioned above can be applied not only to the GP algorithm but also to the P algorithm. Since the variance of the P algorithm does not decrease exponentially with $k$, it can be seen that our algorithm has an advantage over the P algorithm in this regard. We remark that for this particular projector, the decomposition into Hamiltonian powers is costly, and the existence of a better decomposition is possible. Compared to other algorithms, we have provided a numerical benchmark test in the previous section. Additionally, for the CP and RTE algorithms, it is possible to construct approximate projection operators with constant bounded expansion coefficients~\cite{epperly2022, kirby2022}. In this case, further methods are needed to theoretically explain the improvement in measurement cost of our algorithm compared to these algorithms.

\section{Discussions and Conclusions}

We can apply Theorem~\ref{the} to estimate the measurement cost of various quantum KSD algorithms for arbitrary target error $\epsilon$ and subspace dimension $d$.
In this paper, we focus on a target error comparable to the Lanczos algorithm at the same $d$, i.e. $\epsilon\sim\epsilon_K$. 
The Lanczos algorithm is a well-known classical algorithm. It is natural to demand quantum KSD algorithms, as quantum counterparts of Lanczos algorithm, to reach a comparable target error. 
%The Lanczos algorithm is the well-known classical counterpart of quantum KSD algorithms. 
In this case, the measurement cost of the GP algorithm is observed much smaller than the others in the benchmarking. 
%Such results can be explained by composing a projector with a small error $\epsilon_P\sim\epsilon_K$, which holds true for the CP projector as it provides a tight bound for $\epsilon_K$. 

In certain cases of practical applications, there is definite target error epsilon, e.g. the chemical accuracy. In these cases, these target errors are no longer depending on $\epsilon_K$.
According to the Kaniel–Paige convergence theory 
\cite{bjorck2008}, by increasing the subspace dimension $d$ and measurement number $M$, we can always achieve arbitrary finite target error. 
%As illustrated in Fig.~\ref{fig:error_vs_overhead}, the measurement overhead $\gamma$ for all the bases can be small if we aim at a large $\epsilon\gg\epsilon_K$.
Our results focus on relatively small d, in which case the GP algorithm requires a much smaller $M$ than the others, as illustrated in Fig.~\ref{fig:distribution}. We show in Fig.~\ref{fig:scaling} that $M$ converges to $O(1/\epsilon^2)$ when $d$ is sufficiently large, as demonstrated in Refs.~\cite{kirby2022, kirby2024analysis}. Increasing $d$ can weaken the advantage of the GP algorithm in the measurement cost. However, increasing $d$ not only brings more measurement entries (i.e. a larger factor $\beta$), but also implies deeper quantum circuits and more gate noises.

In conclusion, we have proposed a regularised estimator of the ground-state energy (see Algorithm~\ref{alg:QKSD}). Even in the presence of statistical errors, the regularised estimator is variational (i.e.~equal to or above the true ground-state energy) and has a rigorous upper bound. Because of these properties, it provides a universal way of analysing the measurement cost in quantum KSD algorithms, with the result summarised in Theorem~\ref{the}. This approach can be generalised to analyse the effect of other errors, e.g.~imperfect quantum gates. The robustness to statistical errors in our algorithm implies the robustness to other errors. 

To minimise the measurement cost, we have proposed a protocol for constructing Hamiltonian power with an additional Gaussian factor. This Gaussian factor is important because it leads to the statistical error decreasing exponentially with the power. Consequently, we can construct a Chebyshev-type projector onto the ground state with a small measurement cost. We benchmark quantum KSD algorithms with two models of quantum many-body systems. We find that our algorithm requires the smallest measurement number, and a measurement overhead $\gamma$ of a hundred is almost always sufficient for approaching the theoretical-limit accuracy $\epsilon_K$ of the Lanczos algorithm. In addition to the advantage over the classical Lanczos algorithm in scalability, our result suggests that the quantum algorithm is also competitive in accuracy at a reasonable measurement cost. 

This paper primarily focuses on the ground-state problem. Notably, our quantum KSD algorithm holds promise for application to low-lying excited states as well. We leave it for future research.

\begin{acknowledgments}
YL thanks Xiaoting Wang and Zhenyu Cai for the helpful discussions. This work is supported by the National Natural Science Foundation of China (Grants No. 12225507 and No. 12088101), Innovation Program for Quantum Science and Technology (Grant No. 2023ZD0300200), and NSAF (Grants No. U2330201 and No. U2330401). The source codes for the numerical simulation are available at GitHub~\cite{code}.
\end{acknowledgments}

%{\it Note added.}--The numerical results in this paper are different from the first version posted on arXiv. It is due to a typo in the source code for generating Hamiltonians of the Hubbard model. Because of the typo, the $\gamma$ distribution of the F algorithm is significantly changed. When the generated non-Hubbard model is used in the benchmarking, the median $\gamma$ value of the F algorithm is as large as $3\times 10^4$. The code with the typo is available together with other codes.

\appendix

\section{General formalism of KSD algorithm}
\label{app:intro}

Given a Hamiltonian $H$, a reference state $\ket{\varphi}$ and an integer $d$, the standard Krylov subspace is $\mathcal{K} = \mathrm{Span}\left(\ket{\varphi}, H \ket{\varphi}, H^2 \ket{\varphi},\ldots,H^{d-1} \ket{\varphi}\right)$. This can be seen as taking polynomials $f_k(x) = x^{k-1}$. The subspace is the same when $\{f_1,f_2,\ldots,f_d\}$ is an arbitrary set of linearly-independent $(d-1)$-degree polynomials. Any state in the Krylov subspace can be expressed as 
\begin{eqnarray}
\ket{\psi_K} &=& \sum_{k=1}^d a_k f_k(H) \ket{\varphi}.
\end{eqnarray} 
The state that minimises the Rayleigh quotient $\langle \psi_K\vert H\vert \psi_K\rangle/\braket{\psi_K}{\psi_K}$ is regarded as the approximate ground state of the Hamiltonian $H$, and the corresponding energy is the approximate ground-state energy, i.e. 
\begin{eqnarray}
\ket{\psi_g} &\sim & \ket{\psi_{min}} = \operatorname*{arg\,min}_{\psi_K\in \mathcal{K}} \frac{\langle \psi_K\vert H\vert \psi_K\rangle}{\braket{\psi_K}{\psi_K}}, \\
E_g &\sim & E_{min} = \min_{\psi_K\in \mathcal{K}} \frac{\langle \psi_K\vert H\vert \psi_K\rangle}{\braket{\psi_K}{\psi_K}}.
\end{eqnarray}

\begin{definition}[Rayleigh quotient]
For any matrix $\mathbf A\in \mathbb C^{d\times d}$, and any non-zero vector $\mathbf a \in \mathbb C^d$, the Rayleigh quotient is defined as 
\begin{eqnarray}
\mean{\mathbf A}(\mathbf a) &=& \frac{\mathbf a^\dag \mathbf A \mathbf a}{\mathbf a^\dag \mathbf a}.
\label{eq:A}
\end{eqnarray}
\end{definition}

The approximate ground state energy can be rewritten as 
\begin{eqnarray}
E_{min} &=& \min_{\mathbf a\neq 0} \frac{\mathbf a^\dag \mathbf H \mathbf a}{\mathbf a^\dag \mathbf S \mathbf a} = \min_{\mathbf a\neq 0} \frac{\mean{\mathbf H}(\mathbf a)}{\mean{\mathbf S}(\mathbf a)},
\label{eq:Emin}
\end{eqnarray}
where $\mathbf H$ and $\mathbf S$ are $d\times d$ Hermitian matrices as defined in the main text. Notice that $E_{min}$ is the minimum eigenvalue of Eq.~(\ref{eq:eig}).

When $k$ increases, $H^k \ket{\varphi}$ converges to the eigenstate (having non-zero overlap with $\ket{\varphi}$) with the largest absolute eigenvalue. The eigenstate is usually the ground state (If not, take $H\leftarrow H-E_0$ where $E_0$ is a positive constant). Therefore, it is justified to express the ground state as a linear combination of $\{H^k \ket{\varphi}\}$ as long as $\ket{\varphi}$ has a non-zero overlap with the true ground state $\ket{\psi_g}$. However, the convergence with $k$ causes inherent trouble that basis vectors of the Krylov subspace are nearly linearly dependent. Then, the overlap matrix $\bfS$, which is a Gram matrix, becomes nearly singular and has a large condition number. As a result, the denominator of the Rayleigh quotient can be very small, and a tiny error in computing may cause a large deviation of the approximate ground-state energy $E_{min}$. 

According to the Rayleigh-Ritz theorem (Rayleigh quotient theorem)~\cite{horn2012}, the minimisation problem is equivalent to the generalised eigenvalue problem Eq.~(\ref{eq:eig}). The generalised eigenvalue problem can be reduced to an eigenvalue problem while the singularity issue remains. For example, one can use the Cholesky decomposition to express $\bfS$ as the product of an upper-triangular matrix and its conjugate transpose, i.e.~$\bfS = \mathbf R^\dag \mathbf R$, and the eigenvalue problem becomes $\left(\mathbf R^\dag\right)^{-1} \bfH \mathbf R^{-1} (\mathbf R \bfa) = E \mathbf R \bfa$. The Cholesky decomposition, however, also requires the matrix $\bfS$ to be positive-definite. When $\bfS$ is nearly singular, the Cholesky decomposition is unstable. The QZ algorithm is a more numerically stable method, and it is built into most standard programming languages. One way to address the singularity issue is by adding a positive diagonal matrix to $\bfS$. In this work, we also add a diagonal matrix to $\bfH$ to ensure that the Rayleigh quotient is variational, i.e.~the energy is not lower than the ground-state energy (with a controllable probability). An alternative way to address the singularity issue is the so-called thresholding method, see Refs.~\cite{motta2019,epperly2022}. 

\section{Norm distributions of noise matrices}\label{app:norm_distribution}

In this section, we analyse the distributions of $\norm{\hat{\bfH}-\bfH}_2$ and $\norm{\hat{\bfS}-\bfS}_2$. For the two matrices, there are a total of $2d^2$ matrix entries. We can measure each of them independently on the quantum computer. Under this independent measurement protocol, we obtain the most general results that apply to all the bases in Table~\ref{table}, which are summarised in Lemma~\ref{lem:pro} and Lemma~\ref{lem:pro_improved}. In the $2d^2$ matrix entries, some of them are equivalent. To reduce the measurement cost, we can only measure one entry in each set of equivalent matrix entries, and we call it the collective measurement protocol. Because the equivalent-entry sets depend on the basis, the corresponding distributions are basis-dependent. The results are summarised in Lemma~\ref{lem:pro_real} and Lemma~\ref{lem:pro_complex}; we remark that these results are obtained under the assumption that statistical noise is Gaussian.

\begin{lemma}[Based on Chebyshev's inequality]\label{lem:pro}
Let $\kappa$ be the failure probability. The inequality (\ref{eq:condition}) holds when $\eta\sqrt{M} \geq \frac{2d^2}{\sqrt{\kappa}}$.
\end{lemma}

\begin{proof}
Recall the variance upper bounds of $\hat\bfH_{k,q}$ and $\hat\bfS_{k,q}$ given in Eq.~(\ref{eq:varH1}) and (\ref{eq:varS1}), respectively. When $M\geq \frac{4d^4}{\kappa\eta^2}$, 
\begin{eqnarray}
\mathrm{Var}(\hat\bfH_{k,q}) &\leq & \frac{\kappa C_{\bfH}^2 \eta^2}{2d^4}, \label{eq:varH}\\
\mathrm{Var}(\hat\bfS_{k,q}) &\leq & \frac{\kappa C_{\bfS}^2 \eta^2}{2d^4}. \label{eq:varS}
\end{eqnarray}
According to Chebyshev's inequality, we have 
\begin{eqnarray}
\Pr\left(\vert \hat \bfH_{k,q} - \bfH_{k,q} \vert \geq \frac{C_{\bfH}\eta}{d} \right) &\leq & \frac{\kappa}{2d^2},\label{eq:ChebyshevH}\\
\Pr\left(\vert \hat \bfS_{k,q} - \bfS_{k,q} \vert \geq \frac{C_{\bfS}\eta}{d} \right) &\leq & \frac{\kappa}{2d^2}.\label{eq:ChebyshevS}
\end{eqnarray}
	
Since matrix entries are measured independently, estimators $\hat \bfH_{k,q}$ or $\hat \bfS_{k,q}$ are independent. We have
\begin{eqnarray}
&&\Pr \left(\begin{array}{c}
\forall k,q\ \vert \hat \bfH_{k,q} - \bfH_{k,q} \vert \leq \frac{C_{\bfH}\eta}{d} \\
{\rm and}\ \vert \hat \bfS_{k,q} - \bfS_{k,q} \vert \leq \frac{C_{\bfS}\eta}{d}
\end{array}\right) \notag\\
&\geq & \left[\left(1-\frac{\kappa}{2d^2}\right)^{d^2}\right]^2 \geq 1-\kappa,
\end{eqnarray}
where Bernoulli's inequality is used. 
	
Let $\mathbf A$ be an $m\times n$ matrix with entries $A_{i,j}$, its spectral norm satisfies
\begin{eqnarray}
\norm{A}_2 &\leq& \sqrt{mn} \max_{i,j} \vert A_{i,j}\vert.
\end{eqnarray}
This is because of the well known relation $\norm{A}_2 \leq \norm{A}_F$, where the Frobenius norm is $\norm{A}_F=\left(\sum_{i=1}^m \sum_{j=1}^n \abs{A_{i,j}}^2\right)^{1/2}$. Therefore, when $\vert \hat \bfH_{k,q} - \bfH_{k,q} \vert \leq \frac{C_{\bfH}\eta}{d}$ and $\vert \hat \bfS_{k,q} - \bfS_{k,q} \vert \leq \frac{C_{\bfS}\eta}{d}$ for all $k$ and $q$, we have $\norm{ \hat \bfH - \bfH}_2 \leq C_{\bfH}\eta$ and $\norm{ \hat \bfS - \bfS}_2 \leq C_{\bfS}\eta$. 
\end{proof}

The above proof applies to all bases and measurement protocols and holds without any assumption on the statistical noise. As the most general result, Lemma~\ref{lem:pro} is used to prove Theorem~\ref{the} in Appendix~\ref{app:proof}. 

To measure matrix entries on a quantum computer, a practical method is using Hadamard-test circuits. Then, measurement outcomes are independent random variables $\pm 1$, i.e. the random variables are bounded even taking into account other factors (We will give an example later, see Algorithm~\ref{alg:MC}, where the phase of $q_s$ can be absorbed into $U_s$). In this case, we can use Hoeffding's inequality \cite{hoeffding1963} to obtain an improved result.

\begin{lemma}[Based on Hoeffding's inequality]\label{lem:pro_improved}
Suppose matrix entries $\hat \bfH_{k,q}$ and $\hat \bfS_{k,q}$ are empirical means of random variables in ranges $[-C_{\bfH},C_{\bfH}]$ and $[-C_{\bfS},C_{\bfS}]$, respectively. Let $\kappa$ be the failure probability. If the matrix elements of $\bfH$ and $\bfS$ are estimated using Algorithm~\ref{alg:MC}, then the inequality (\ref{eq:condition}) holds when $\eta\sqrt{M}\geq \sqrt{2d^2\ln{\frac{8d^2}{\kappa}}}$. 
\end{lemma}

\begin{proof}
%Using Algorithm~\ref{alg:MC}, the real part (imaginary part) of $\bfH_{k,q}$ can be estimated as the sum of $C_{\bfH}\mu_l^X/M$ ($C_{\bfH}\mu_l^Y/M$). (Add a rotation gate $R_Z[\arg(q_s)]=e^{-iZ\arg(q_s)/2}$ before measurement.) 
For each matrix entry, its real (imaginary) part is computed by taking the mean of $M$ random variables.
According to Hoeffding's inequality, we have, 
\begin{eqnarray}
&&\Pr\left(\vert \hat \bfH_{k,q} - \bfH_{k,q} \vert \geq \frac{C_{\bfH}\eta}{d} \right)\notag\\
&\leq&\Pr\left(\vert \Re(\hat \bfH_{k,q}) - \Re(\bfH_{k,q}) \vert \geq \frac{C_{\bfH}\eta}{\sqrt{2}d} \right)\notag\\
&\times&\Pr\left(\vert \Im(\hat \bfH_{k,q}) - \Im(\bfH_{k,q}) \vert \geq \frac{C_{\bfH}\eta}{\sqrt{2}d} \right)\notag\\
&\leq&4e^{-\frac{M\eta^2}{2d^2}}.
\end{eqnarray}
Similarly,
\begin{eqnarray}
\Pr\left(\vert \hat \bfS_{k,q} - \bfS_{k,q} \vert \geq \frac{C_{\bfS}\eta}{d} \right) \leq 4e^{-\frac{M\eta^2}{2d^2}}.
\end{eqnarray}

When $M\geq \frac{2d^2}{\eta^2}\ln{\frac{8d^2}{\kappa}}$, Eq.~(\ref{eq:ChebyshevH}) and Eq.~(\ref{eq:ChebyshevS}) hold. The rest of the proof is the same as in Lemma~\ref{lem:pro}. 
\end{proof}

Considering the equivalence between matrix entries, we can find that $\bfH$ and $\bfS$ are real Hankel matrices for the GP and P, IP and ITE bases in Table~\ref{table}, real symmetric matrices for the CP and F bases and complex Hermitian-Toeplitz matrices for the RTE basis. Under the collective measurement protocol, random matrices $\hat{\bfH}-\bfH$ and $\hat{\bfS}-\bfS$ are matrices of the same type as $\bfH$ and $\bfS$. To analyse the distribution of these random matrices, we assume that the distributions of $\hat{\bfH}_{k,q}-\bfH_{k,q}$ and $\hat{\bfS}_{k,q}-\bfS_{k,q}$ are Gaussian with the zero mean. Notice that inequivalent entries are independent random variables. Then, we obtain the following results. 

\begin{lemma}[Real symmetric/Hankel matrices]\label{lem:pro_real}
Let $\kappa$ be the failure probability. If $\hat{\bfH}-\bfH$ and $\hat{\bfS}-\bfS$ are random real symmetric matrices following the Gaussian distribution, then the inequality (\ref{eq:condition}) holds when $\eta\sqrt{M}\geq \sqrt{2d\ln{\frac{4d}{\kappa}}}$.
\end{lemma}
\begin{proof}
Applying the matrix Gaussian series theorem~\cite{tropp2015} to the real symmetric matrix with Gaussian noise, we have 
\begin{eqnarray}
\Pr\left(\norm{\hat{\bfH}-\bfH}_2\geq C_{\bfH}\eta\right)\leq 2de^{-\frac{M\eta^2}{2d}},
\end{eqnarray}
where ${\rm Var}(\hat{\bfH}_{k,q}) \leq \frac{C_{\bfH}^2}{M}$ is used. Similarly,
\begin{eqnarray}
\Pr\left(\norm{\hat{\bfS}-\bfS}_2\geq C_{\bfS}\eta\right)\leq 2de^{-\frac{M\eta^2}{2d}}.
\end{eqnarray}
Note that the real Hankel matrices with Gaussian noise yield the same inequalities \cite{raymundo2021}.
	
To make sure $\norm{ \hat \bfH - \bfH}_2 \leq C_{\bfH}\eta$ and $\norm{ \hat \bfS - \bfS}_2 \leq C_{\bfS}\eta$ with a probability of at least $1-\kappa$, we take $M$ such that 
\begin{eqnarray}
2d e^{-\frac{M\eta^2}{2d}} = \frac{\kappa}{2}.
\end{eqnarray}
Then, 
\begin{eqnarray}
M=\frac{2d}{\eta^2}\ln{\frac{4d}{\kappa}}.
\end{eqnarray}
\end{proof}

\begin{lemma}[Complex Hermitian-Toeplitz matrices]\label{lem:pro_complex}
Let $\kappa$ be the failure probability. If $\hat{\bfH}-\bfH$ and $\hat{\bfS}-\bfS$ are random complex Hermitian-Toeplitz matrices following the Gaussian distribution, then the inequality (\ref{eq:condition}) holds when $\eta\sqrt{M}\geq \sqrt{2(2d-1)\ln{\frac{4d}{\kappa}}}$. 
\end{lemma}
\begin{proof}
For complex Hermitian-Toeplitz matrices, the matrix Gaussian series theorem~\cite{tropp2015} gives
\begin{eqnarray}
\Pr\left(\norm{\hat{\bfH}-\bfH}_2\geq C_{\bfH}\eta\right)\leq 2de^{-\frac{M\eta^2}{2(2d-1)}}.
\end{eqnarray}
Similarly,
\begin{eqnarray}
\Pr\left(\norm{\hat{\bfS}-\bfS}_2\geq C_{\bfS}\eta\right)\leq 2de^{-\frac{M\eta^2}{2(2d-1)}}.
\end{eqnarray}
	
To make sure $\norm{ \hat \bfH - \bfH}_2 \leq C_{\bfH}\eta$ and $\norm{ \hat \bfS - \bfS}_2 \leq C_{\bfS}\eta$ with a probability of at least $1-\kappa$, we take
\begin{eqnarray}
M=\frac{2(2d-1)}{\eta^2}\ln{\frac{4d}{\kappa}}.
\end{eqnarray}
\end{proof}

%In this section, we have used the random matrix theory to analyse the distributions of $\norm{\hat{\bfH}-\bfH}_2$ and $\norm{\hat{\bfS}-\bfS}_2$. To satisfy Eq.~(\ref{eq:condition}), the measurement number $M$ has a lower bound, which in turn gives the regularisation parameter $\eta$ in Algorithm~\ref{alg:QKSD}. Here, $\eta$ is the function of the measurement number $M$, the failure probability $\kappa$ and the Krylov subspace $d$. Given $\eta$, which is the solution to Eq.~(\ref{eq:equation}), the corresponding measurement number $M$ can also be identified. In Table~\ref{tab:eta}, we list $\eta$ and $M$ for different cases.

\section{Details on the analysis of measurement cost}
\label{app:proof}

In this section, we develop theoretical tools for analysing the measurement cost in quantum KSD algorithms. At the end of this section, we give the proofs of Lemma~\ref{lem} and Theorem~\ref{the}.

\begin{definition}
For simplification, we define the following functions: 
\begin{eqnarray}
E(\bfa) &=& \frac{\mean{\bfH}(\bfa)}{\mean{\bfS}(\bfa)}, \label{eq:def1}\\
\hat E(\eta, \bfa) &=& \frac{\mean{\hat\bfH}(\bfa) + C_{\bfH} \eta}{\mean{\hat\bfS}(\bfa) + C_{\bfS} \eta}, \label{eq:def2}\\
E'(\eta, \bfa) &=& \frac{\mean{\bfH}(\bfa)+ 2C_{\bfH} \eta}{\mean{\bfS}(\bfa)+ 2C_{\bfS} \eta}.\label{eq:def3}
\end{eqnarray}
\end{definition}
Here, functions $\mean{\mathbf A}(\mathbf a)$ are defined in Eq. (\ref{eq:A}). $\min_{\mathbf a\neq 0}E(\bfa) = E_{min}$ is the minimum eigenvalue of the generalised eigenvalue problem in Eq.~(\ref{eq:eig}); $\min_{\mathbf a\neq 0}\hat E(\eta, \bfa) = \hat E_{min}$ is the minimum eigenvalue of the generalised eigenvalue problem solved in Algorithm~\ref{alg:QKSD}; and Eq.~(\ref{eq:def3}) is equivalent to Eq.~(\ref{eq:Ep}).

\begin{lemma}\label{lem:statements}
If $\norm{ \hat \bfH - \bfH}_2 \leq C_{\bfH}\eta $ and $\norm{ \hat \bfS - \bfS}_2 \leq C_{\bfS}\eta$, the following statements hold: 
\begin{enumerate}
\item[(1)] $\hat E(\eta, \bfa) \geq \min\{E(\bfa) ,0\}$; 
\item[(2)] If $\hat E(\eta, \bfa)<0$, then $\hat E(\eta,\bfa)\leq E^\prime(\eta,\bfa)$; 
\item[(3)] If $ E'(\eta,\bfa)<0$, then $\hat E(\eta,\bfa)<0$. 
\end{enumerate}
\end{lemma}

\begin{proof}
Consider the error in the denominator, 
\begin{eqnarray}
&& \abs{\mean{\hat \bfS}(\bfa) - \mean{\bfS}(\bfa)}
= \absLR{\frac{\bfa^\dag \hat \bfS\bfa - \bfa^\dag \bfS \bfa}{\bfa^\dag \bfa}} \notag\\
&& = \absLR{\frac{\bfa^\dag (\hat \bfS-\bfS)\bfa }{\bfa^\dag \bfa}}
\leq \frac{\abs{\bfa^\dag } \norm{\hat \bfS-\bfS}_2 \abs{\bfa}}{\bfa^\dag \bfa} \notag\\
&& \leq C_{\bfS} \eta,
\end{eqnarray}
from which we obtain 
\begin{eqnarray}\label{eq:meanS}
0 < \mean{\bfS}(\bfa) &\leq& \mean{\hat\bfS}(\bfa) + C_{\bfS} \eta \notag\\
&\leq& \mean{\bfS}(\bfa) + 2C_{\bfS} \eta.
\end{eqnarray}
Similarly, we have 
\begin{eqnarray}\label{eq:meanH}
\mean{\bfH}(\bfa) &\leq& \mean{\hat\bfH}(\bfa) + C_{\bfH} \eta \notag\\
&\leq& \mean{\bfH}(\bfa) + 2C_{\bfH} \eta.
\end{eqnarray}
	
Let $a$, $b$, $c$ and $d$ be positive numbers, it can be shown that 
\begin{eqnarray}\label{eq:frac_ineq}
\frac{-a}{b} &<& \frac{-a+c}{b+d}.
\end{eqnarray}
If $\hat{E}(\eta,\bfa)<0$, then $\mean{\hat\bfH}(\bfa) + C_{\bfH} \eta<0$. Under this condition, by using Eqs.~(\ref{eq:meanS})-(\ref{eq:frac_ineq}), we get 
\begin{eqnarray}\label{eq:quotients}
\frac{\mean{\bfH}(\bfa)}{\mean{\bfS}(\bfa)} &\leq&
\frac{\mean{\hat\bfH}(\bfa) + C_{\bfH} \eta}{\mean{\hat\bfS}(\bfa) + C_{\bfS} \eta} \notag\\
&\leq& \frac{\mean{\bfH}(\bfa)+ 2C_{\bfH} \eta}{\mean{\bfS}(\bfa)+ 2C_{\bfS} \eta},
\end{eqnarray}
i.e.~$E(\bfa) \leq \hat{E}(\eta,\bfa) \leq E'(\eta,\bfa)$. Both the first and second statements are proven.

If $E'(\eta, \bfa)<0$, then $\mean{\hat\bfH}(\bfa) + C_{\bfH} \eta \leq \mean{\bfH}(\bfa)+ 2C_{\bfH} \eta <0$. Therefore, $\hat E(\eta,\bfa)<0$. The third statement is proven. 
\end{proof}

Under the condition $E'(\eta,\bfa)<0$, Lemma~\ref{lem:statements} states $E(\bfa) \leq \hat{E}(\eta,\bfa) \leq E'(\eta,\bfa)$, where the first inequality shows that the algorithm is variational, and the second inequality gives the error bound. In the following, we will show
that this relation still holds after minimisation.

\begin{lemma}\label{lem:minE}
Under conditions 
\begin{enumerate}
\item[(1)] $\norm{\hat{\bfH}-\bfH}_2 \leq C_{\bfH}\eta$ and $\norm{ \hat{\bfS}-\bfS}_2 \leq C_{\bfS}\eta$, and 
\item[(2)] $E_g<0$ and $\min_{\bfa\neq \bf{0}} E'(\eta,\bfa)<0$,
\end{enumerate}
the following statement holds, 
\begin{eqnarray}
E_g \leq \min_{\bfa\neq \bf{0}} \hat{E}(\eta,\bfa) \leq \min_{\bfa\neq \bf{0}} E'(\eta,\bfa). 
\end{eqnarray}
\end{lemma}

\begin{proof}
According to the first statement in Lemma~\ref{lem:statements}, $\forall \bfa\neq \bf{0}$, $\hat E(\eta,\bfa) \geq \min\{E(\bfa) ,0\}$. Since $E(\bfa)\geq E_g$ and $0>E_g$, we have $\min_{\bfa \neq \bf{0}} \hat{E}(\eta,\bfa) \geq E_g$. 
	
Let $\bfa^*=\operatorname*{arg\,min}_{\bfa\neq \bf{0}} E'(\eta,\bfa)$. According to the third statement in Lemma~\ref{lem:statements}, the condition $E'(\eta,\bfa^*)<0$ implies that $\hat{E}(\eta,\bfa^*)<0$. Under this condition, according to the second statement in Lemma~\ref{lem:statements}, $\hat{E}(\eta, \bfa^*)\leq E'(\eta, \bfa^*)$. Therefore, we have 
\begin{eqnarray}
&& \min_{\bfa\neq \bf{0}} \hat{E}(\eta,\bfa) \leq \hat{E}(\eta, \bfa^*) \notag \\
&\leq & E'(\eta, \bfa^*) = \min_{\bfa\neq \bf{0}} E'(\eta,\bfa).  
\end{eqnarray}
\end{proof}

\begin{comment}
\begin{lemma}\label{lem:solution}
Under the condition $\min_{\bfa\neq 0} E'(\eta,\bfa)<0$, when $\epsilon>\min_{\bfa\neq 0} E(\bfa)-E_g$, the equation $\min_{\bfa\neq 0} E'(\eta, \bfa) = E_g + \epsilon$ has a positive solution, and the solution is unique.
\end{lemma}

\begin{proof}
Let $\bfa^*=\operatorname*{arg\,min}_{\bfa\neq 0} E'(\eta,\bfa)$, and let $\eta'<\eta$ be any positive numbers. It is obvious that $\min_{\bfa\neq 0} E'(\eta',\bfa) \leq E'(\eta', \bfa^*) < E'(\eta, \bfa^*) = \min_{\bfa\neq 0} E'(\eta,\bfa) < 0$. Therefore, under the condition $\min_{\bfa\neq 0} E'(\eta,\bfa)<0$, $\min_{\bfa\neq 0} E'(\eta,\bfa)$ is a continuous monotonically increasing function of $\eta$ when $\eta\geq 0$. When $\eta=0$, $\min_{\bfa\neq 0}E'(0,\bfa)=\min_{\bfa\neq 0}E(\bfa)$. Therefore, for all $E_g + \epsilon>\min_{\bfa\neq 0}E(\bfa)$, the equation has a positive solution, and the solution is unique. 
\end{proof}
\end{comment}

\begin{lemma}\label{lem:solution}
Suppose $\bfS$ is invertible, $\epsilon>\epsilon_K$ and $E_g+\epsilon<0$. Eq.~(\ref{eq:equation}) has a positive solution in $\eta$, and the solution is unique.
\end{lemma}

\begin{proof}
Let $\bfa^*_\eta=\operatorname*{arg\,min}_{\bfa\neq \bf{0}} E'(\eta,\bfa)$. 

First, we prove that $\min_{\bfa\neq \bf{0}} E'(\eta,\bfa)$ is a strictly monotonically increasing function of $\eta$ when $\eta$ is in the range $\{\eta\geq 0\vert\min_{\bfa\neq \bf{0}} E'(\eta,\bfa)<0\}$. Let $\eta_2$ and $\eta_1$ be in the range and $\eta_2>\eta_1$. By Eq.~(\ref{eq:frac_ineq}), we have $E'(\eta_1,\bfa^*_{\eta_2})<E'(\eta_2,\bfa^*_{\eta_2})$. Since $\min_{\bfa\neq \bf{0}} E'(\eta_1,\bfa)\leq E'(\eta_1,\bfa^*_{\eta_2})$, 
\begin{eqnarray}
\min_{\bfa\neq \bf{0}} E'(\eta_1,\bfa) < \min_{\bfa\neq \bf{0}} E'(\eta_2,\bfa).
\end{eqnarray}

Second, we prove that $\min_{\bfa\neq \bf{0}} E'(\eta,\bfa)$ is a continuous function of $\eta$ when $\eta\geq 0$. When $\eta\geq 0$ and $\bfa\neq \bf{0}$, $E'(\eta,\bfa)$ is a continuous function of $\eta$. Consequently, $\forall \delta_E>0$, $\exists \delta_\eta>0$ such that $E'(\eta+\delta_\eta,\bfa^*_\eta) - E'(\eta,\bfa^*_\eta)<\delta_E$. Then, 
\begin{eqnarray}
0 &<& E'(\eta+\delta_\eta,\bfa^*_{\eta+\delta_\eta})-E'(\eta,\bfa^*_\eta) \notag \\
&<& E'(\eta+\delta_\eta,\bfa^*_{\eta})-E'(\eta,\bfa^*_\eta) \leq \delta_E. 
\end{eqnarray}
Therefore, $\forall \delta_E>0$, $\exists \delta_\eta>0$ such that 
\begin{eqnarray}
0 &<& \min_{\bfa\neq \bf{0}} E'(\eta+\delta_\eta,\bfa)-\min_{\bfa\neq \bf{0}} E'(\eta,\bfa) \notag\\
&<& \delta_E.
\end{eqnarray}

Combining the two, we conclude that $\min_{\bfa\neq \bf{0}} E'(\eta,\bfa)$ is a strictly monotonically increasing continuous function of $\eta$ when $\eta\geq 0$ and $\min_{\bfa\neq \bf{0}} E'(\eta,\bfa)<0$. Taking $\eta = 0$, we have $\min_{\bfa\neq \bf{0}} E'(0,\bfa) = E_g+\epsilon_K$. Therefore, when $\epsilon>\epsilon_K$ and $E_g+\epsilon<0$, the equation $\min_{\bfa\neq \bf{0}} E'(\eta,\bfa)=E_g+\epsilon$ has a unique solution. 
\end{proof}

\subsection{Proof of Lemma~\ref{lem}}

\begin{proof}
For a regularisation parameter $\eta$ satisfying Eq.~(\ref{eq:condition}), the first condition of Lemma~\ref{lem:minE} holds up to the failure probability $\kappa$. Under conditions $E_g<0$ and $\min_{\bfa\neq \bf{0}} E'(\eta,\bfa)<0$, according to Lemma~\ref{lem:minE}, $E_g \leq \hat{E}_{min} \leq \min_{\bfa\neq \bf{0}} E'(\eta,\bfa)$ holds up to the failure probability $\kappa$. Eq.~(\ref{eq:hatE_bound}) is proven.
\end{proof}

\subsection{Proof of Theorem~\ref{the}}

\begin{proof}
Suppose $\epsilon>\epsilon_K$ and $E_g+\epsilon<0$. The existence of the solution $\eta$ is proved in Lemma~\ref{lem:solution}. With the solution $\eta$, Eq.~(\ref{eq:equation}) is satisfied. 
	
If we take the measurement number
\begin{eqnarray}
M = \frac{4d^4}{\kappa\eta^2},
\end{eqnarray}
$\norm{ \hat \bfH - \bfH}_2 \leq C_{\bfH}\eta $ and $\norm{ \hat \bfS - \bfS}_2 \leq C_{\bfS}\eta$ hold up to the failure probability $\kappa$ as proved in Lemma~\ref{lem:pro}. Then, according to Lemma~\ref{lem:minE}, $E_g \leq \min_{\bfa\neq \bf{0}} \hat{E}(\eta,\bfa) \leq \min_{\bfa\neq \bf{0}} E'(\eta, \bfa) = E_g + \epsilon$ holds up to the failure probability. 
	
In the independent measurement protocol, the total number of measurements is 
\begin{eqnarray}\label{eq:Mtot1}
M_{tot} &=& M\times 2 \times d^2 \times 2 = \frac{16d^6}{\kappa\eta^2}.
\end{eqnarray}
Notice that for each matrix entry $\hat \bfH_{k,q}$ or $\hat \bfS_{k,q}$, we need $M$ measurements for its real part and $M$ measurements for its imaginary part. There are two matrices, and each matrix has $d^2$ entries. To match Eq.~(\ref{eq:Mtot1}), we can take $\alpha(\kappa)=256/\kappa$ and $\beta(d)=d^6$ in Eq.~(\ref{eq:Mtot}).
\end{proof}

In the above proof, we have used Lemma~\ref{lem:pro} to give the most general and conservative estimation of the measurement cost. In fact, according to Appendix~\ref{app:norm_distribution}, when considering details of the quantum KSD algorithms, $M$ (and $M_{tot}$) can be optimised, so as $\alpha(\kappa)$ and $\beta(d)$; see Appendix~\ref{app:optimisation}.

\subsection{Remarks on the measurement number}

First, in our algorithm (i.e.~GP) and some other KSD algorithms (i.e.~P, CP, IP, ITE and F), matrices $\bfH$ and $\bfS$ are real. In this case, we only need to measure the real part. The cost is directly reduced by a factor of two. Only measuring the real part also reduces variances by a factor of two, i.e.~from $\mathrm{Var}(\hat{\bfH}_{k,q})\leq 2C_{\bfH}^2/M$ and $\mathrm{Var}(\hat{\bfS}_{k,q})\leq 2C_{\bfS}^2/M$ to $\mathrm{Var}(\hat{\bfH}_{k,q})\leq C_{\bfH}^2/M$ and $\mathrm{Var}(\hat{\bfS}_{k,q})\leq C_{\bfS}^2/M$. Because of the reduced variance, the cost is reduced by another factor of two. Overall, the measurement cost is reduced by a factor of four in algorithms using real matrices. 

Second, when the failure probability $\kappa$ is low, typically $\norm{ \hat \bfH - \bfH}_2 \ll C_{\bfH}\eta$ and $\norm{ \hat \bfS - \bfS}_2 \ll C_{\bfS}\eta$. In this case, the error is overestimated in $E'$, and the typical error is approximately given by 
\begin{eqnarray}
E''(\eta,\bfa) = \frac{\bfa^\dag (\bfH+C_{\bfH}\eta) \bfa}{\bfa^\dag (\bfS+C_{\bfS}\eta) \bfa}.
\end{eqnarray}
Notice that a factor of two has been removed from the denominator and numerator compared with $E'$. If we consider the typical error $\epsilon = E'' - E_g$ instead of the error upper bound $\epsilon = E' - E_g$, the required sampling cost is reduced by a factor of four. 

\section{Cost for measuring the energy with an ideal projection}
\label{app:projection}

We consider the ideal case that we can realise an ideal projection onto the true ground state. To use the above results, we assume that $f_1=\ketbra{\psi_g}{\psi_g}$ is the projection operator and take $d = 1$ in the KSD algorithm. Then, we have 
\begin{eqnarray}
\bfH_{1,1} &=& \bra{\varphi} f_1^\dag H f_1\ket{\varphi} = p_g E_g, \\
\bfS_{1,1} &=& \bra{\varphi} f_1^\dag f_1\ket{\varphi} = p_g,
\end{eqnarray}
and $E_{min} = \bfH_{1,1}/\bfS_{1,1} = E_g$, i.e.~$\epsilon_K = 0$. 

We suppose $C_{\bfH}=\norm{H}_2$ and $C_{\bfS}=1$. When we take $\eta = \frac{p_g \epsilon}{4\norm{H}_2}$, $\min_{\bfa\neq \bf{0}} E'(\eta, \bfa) \leq E_g + \epsilon$. Notice that $\min_{\bfa\neq \bf{0}} E'(\eta, \bfa) = (\bfH_{1,1}+2\norm{H}_2\eta)/(\bfS_{1,1}+2\eta) \leq E_g + 4\norm{H}_2\eta/p_g$. Accordingly, the total measurement number in the ideal case has an upper bound 
\begin{eqnarray}
M_{tot} = \frac{\alpha(\kappa)\norm{H}_2^2}{p_g^2 \epsilon^2},
\end{eqnarray}
which is the first factor in Eq.~(\ref{eq:Mtot_factoring}). 

\section{Optimised $\alpha$ and $\beta$ factors}
\label{app:optimisation}

Theorem~\ref{the} is proved taking the bound given by Lemma~\ref{lem:pro}. However, Lemma~\ref{lem:pro} aims to give a general result, leading to a conservative estimation of $\alpha$ and $\beta$. In this section, based on other results in Appendix~\ref{app:norm_distribution}, we will show that these two factors can be further reduced. We summarise the reduced $\alpha$ and $\beta$ in Table~\ref{tab:eta} as well. 

For the independent measurement protocol, Lemma~\ref{lem:pro_improved} guarantees that the lower bound of the total measurement number is 
\begin{eqnarray}
M_{tot}=4d^2M=\frac{8d^4}{\eta^2}\ln{\frac{8d^2}{\kappa}}.
\end{eqnarray}
Assuming $\kappa\ll \frac{1}{8d^2}$ and neglecting $8d^2$ in logarithm yields $\alpha \approx 128\ln\frac{1}{\kappa}$ and $\beta \approx d^4$. 

\begin{definition}
Let $\mathbf{g}\in \mathbb{R}^{2d-1}$ and $\mathbf{G}\in \mathbb{R}^{d\times d}$. If $\mathbf{G}_{i,j} = \mathbf{g}_{i+j-1}$, then $\mathbf{G}$ is a $d\times d$ real Hankel matrix.
\end{definition}

In our algorithm (i.e.~GP) and some other KSD algorithms (i.e.~P, IP and ITE), $\bfH$ and $\bfS$ are real Hankel matrices. The collective measurement protocol only requires measuring $2d-1$ matrix entries for each matrix to construct $\hat{\bfH}$ and $\hat{\bfS}$. Specifically, we measure $\bfH_{\lceil\frac{l}{2}\rceil,\lfloor\frac{l}{2}\rfloor}$ and $\bfS_{\lceil\frac{l}{2}\rceil,\lfloor\frac{l}{2}\rfloor}$ with $l = 2,3,\ldots,2d$. Then we take $\bfH_{i,j} = \bfH_{\lceil\frac{i+j}{2}\rceil,\lfloor\frac{i+j}{2}\rfloor}$ and $\bfS_{i,j} = \bfS_{\lceil\frac{i+j}{2}\rceil,\lfloor\frac{i+j}{2}\rfloor}$ for all $i$ and $j$. According to Lemma~\ref{lem:pro_real}, the total measurement number is 
\begin{eqnarray}
M_{tot} &=& M \times (2d-1) \times 2 \notag\\
&=& \frac{4d(2d-1)}{\eta^2}\ln{\frac{4d}{\kappa}}.
\end{eqnarray}
Assuming $\kappa\ll \frac{1}{4d}$ and neglecting $4d$ in logarithm yields $\alpha = 64\ln\frac{1}{\kappa}$ and $\beta = d(2d-1)$. 

For the CP and F algorithms, $\bfH$ and $\bfS$ are real symmetric matrices. We need to measure $d(d+1)/2$ matrix entries to construct $\hat{\bfH}$ and $\hat{\bfS}$, respectively.  
According to Lemma~\ref{lem:pro_real}, the total measurement number is 
\begin{eqnarray}
M_{tot} &=& M \times d(d+1)/2 \times 2 \notag\\
&=& \frac{2d^2(d+1)}{\eta^2}\ln{\frac{4d}{\kappa}}.
\end{eqnarray}
Assuming $\kappa\ll \frac{1}{4d}$ and neglecting $4d$ in logarithm yields $\alpha = 32\ln\frac{1}{\kappa}$ and $\beta = d^2(d+1)$.  

The RTE algorithm is the only one having complex matrices $\bfH$ and $\bfS$. Considering its Hermitian–Toeplitz structure, we need to measure $d$ real and $d-1$ imaginary matrix entries to estimate $\bfH$ and $\bfS$, respectively. 
According to Lemma~\ref{lem:pro_complex}, the total measurement number is 
\begin{eqnarray}
M_{tot} &=& M \times (2d-1) \times 2 \notag\\
&=& \frac{4(2d-1)^2}{\eta^2}\ln{\frac{4d}{\kappa}}.
\end{eqnarray}
Assuming $\kappa\ll \frac{1}{4d}$ and neglecting $4d$ in logarithm yields $\alpha = 64\ln\frac{1}{\kappa}$ and $\beta = (2d-1)^2$.

Finally, the optimised $\alpha$ and $\beta$ are listed in Table~\ref{tab:eta}. Compared with cases, Hankel matrices have smaller $\alpha$ and $\beta$. For the GP algorithm, taking $\alpha=16\ln\frac{1}{\kappa}$ and $\beta=d(2d-1)$ is sufficient.

\section{Gaussian-power basis}
\label{app:GPbasis}

In this section, we prove Eq.~(\ref{eq:norm_bound}) and Eq.~(\ref{eq:LCU}), respectively. 

\subsection{Norm upper bound}

\begin{lemma}\label{lem:norm}
The Gaussian-power basis operators $f_k$ defined in Eq.~(\ref{eq:fk}) satisfy Eq.~(\ref{eq:norm_bound}).
\end{lemma}
\begin{proof}
By the spectral decomposition of the operator $(H-E_0)\tau$, we have
\begin{eqnarray}\label{eq:spec_decomp}
\tau^{k-1}\norm{f_k}_2 \leq \max_{y\in\mathbb{R}} \abs{y^{k-1} e^{-\frac{1}{2}y^2}},    
\end{eqnarray}
where $y$ corresponds to eigenvalues of $(H-E_0)\tau$. It can be found that
\begin{eqnarray}\label{eq:spec_max}
\max_{y\in\mathbb{R}} \abs{y^{k-1} e^{-\frac{1}{2}y^2}} = \left( \frac{k-1}{e} \right)^{\frac{k-1}{2}}.
\end{eqnarray} 
Combining Eqs.~(\ref{eq:spec_decomp}) and (\ref{eq:spec_max}), Eq.~(\ref{eq:norm_bound}) is proved. 
\end{proof}

\subsection{The integral}

The generating function for the Hermite polynomials is \cite{arfken2013}
\begin{eqnarray}\label{eq:generating_function}
\sum_{n=0}^{\infty}H_n(u)\frac{t^n}{n!}=e^{2ut-t^2}.
\end{eqnarray}

\begin{lemma}
\begin{eqnarray}\label{eq:integral}
\int_{-\infty}^{+\infty}du H_n(u) e^{-u^2-iyu}=\sqrt{\pi}(-iy)^n e^{-\frac{1}{4}y^2}    
\end{eqnarray}
holds for arbitrary $y\in \mathbb{R}$.
\end{lemma}
\begin{proof}
Denote $I_n=\int_{-\infty}^{+\infty}du H_n(u) e^{-u^2-iyu}$. Taking Eq.~(\ref{eq:generating_function}) into account, we have \cite{babusci2012}
\begin{eqnarray}\label{eq:gaussian_integral}
\sum_{n=0}^{\infty}I_n\frac{t^n}{n!}&=&\int_{-\infty}^{+\infty}du e^{-u^2+(2t-iy)u-t^2}\notag\\
&=&\sqrt{\pi}e^{-\frac{1}{4}y^2-iyt}\notag\\
&=&\sqrt{\pi}e^{-\frac{1}{4}y^2}\sum_{n=0}^{\infty}(-iy)^n\frac{t^n}{n!},
\end{eqnarray}
where the Gaussian integral is used. From Eq.~(\ref{eq:gaussian_integral}), we conclude that $ I_n=\sqrt{\pi}(-iy)^n e^{-\frac{1}{4}y^2}$. 
\end{proof}

\begin{lemma}\label{lemma:LCU}
Eq.~(\ref{eq:LCU}) is true. 
\end{lemma}

\begin{proof}
Taking $y=\sqrt{2}x\tau$ and $u=\frac{t}{\sqrt{2}\tau}$ in Eq.~(\ref{eq:integral}), Eq.~(\ref{eq:LCU}) is proved directly. 
\end{proof}

\section{Algorithm and variance}
\label{app:algorithm}

In this section, first, we review the zeroth-order leading-order-rotation formula~\cite{yang2021}, then we analyse the variance and work out the upper bounds of the cost, finally, we give the pseudocode of our measurement-efficient algorithm. 

\subsection{Zeroth-order leading-order-rotation formula}

Assume that the Hamiltonian is expressed as a linear combination of Pauli operators, i.e. 
\begin{eqnarray}
H = \sum_j h_j\sigma_j.
\end{eqnarray}
The Taylor expansion of the time evolution operator is 
\begin{eqnarray}
e^{-iH\Delta t} = \openone-iH\Delta t + T_0(\Delta t),
\end{eqnarray}
where the summation of high-order terms is 
\begin{eqnarray}
T_0(\Delta t) &=& \sum_{k=2}^\infty \sum_{j_1,\ldots,j_k} \frac{\prod_{a=1}^k (-ih_{j_a}\Delta t)}{k!} \notag\\
&& \times \sigma_{j_k}\cdots\sigma_{j_1}.
\end{eqnarray}
The leading-order terms can be expressed as a linear combination of rotation operators, 
\begin{eqnarray}
\openone-iH\Delta t = \sum_j \beta_j(\Delta t) e^{-i\sgn(h_j)\phi(\Delta t)\sigma_j},
\end{eqnarray}
where $\phi(\Delta t) = \arctan(h_{tot}\Delta t)$, $\beta_j(\Delta t) = \abs{h_j}\Delta t/\sin\phi(\Delta t)$ and $h_{tot} = \sum_j \abs{h_j}$. The zeroth-order leading-order-rotation formula is 
\begin{eqnarray}
e^{-i H \Delta t} &=& \sum_j \beta_j(\Delta t) e^{-i\sgn(h_j)\phi(\Delta t)\sigma_j} \notag\\
&& + T_0(\Delta t),
\label{eq:LOR1}
\end{eqnarray}
which is a linear combination of rotation and Pauli operators. Accordingly, for the evolution time $t$ and time step number $N$, the LCU expression of the time evolution operator is 
\begin{eqnarray}
e^{-iHt} &=& \left[\sum_j \beta_j(t/N) e^{-i\sgn(h_j)\phi(t/N)\sigma_j} \right.\notag\\
&&\left. + T_0(t/N)\right]^N.
\label{eq:LOR2}
\end{eqnarray}
The above equation is the explicit form of the expression $e^{-iHt} = \left[\sum_r v_r(t/N) V_r(t/N)\right]^N$ referred in the main text. 

Now, we consider the cost factor, i.e.~the 1-norm of coefficients in an LCU expression. For Eq.~(\ref{eq:LOR1}), the corresponding cost factor is 
\begin{eqnarray}
&&c(\Delta t) = \sum_j \abs{\beta_j(\Delta t)} + \sum_{k=2}^\infty \sum_{j_1,\ldots,j_k} \frac{\prod_{a=1}^k \abs{h_{j_a}\Delta t}}{k!} \notag \\
&&= \sqrt{1+h_{tot}^2\Delta t^2} + e^{h_{tot}\abs{\Delta t}} - (1+h_{tot}\abs{\Delta t}).
\end{eqnarray}
Accordingly, the cost factor of Eq.~(\ref{eq:LOR2}) is $[c(t/N)]^N$. 

\begin{lemma}\label{lemma:cost_LOR}
\begin{eqnarray}
c(\Delta t) \leq e^{\frac{e}{2} h_{tot}^2 \Delta t^2}.
\end{eqnarray}
\end{lemma}

\begin{proof}
Let $x=h_{tot} \abs{\Delta t}$, thus $x\geq 0$. Then 
\begin{eqnarray}
c(\Delta t) &=& \sqrt{1+x^2} + e^{x} - (1+x)\notag\\
&\leq& e^{\frac{x^2}{2}}+ e^{x} - (1+x).
\end{eqnarray}
Define the function 
\begin{eqnarray}
y(x)=e^{\frac{e}{2}x^2}+(1+x)-(e^{x}+e^{\frac{x^2}{2}}).    
\end{eqnarray}
It follows that $y(0)=0$ and 
\begin{eqnarray}
y'(x)&=&e x e^{\frac{e}{2}x^2} +1-(e^{x}+x e^{\frac{x^2}{2}})\notag\\
&\geq& (e-1)x e^{\frac{e}{2}x^2} +1-e^x\label{eq:dy}\\
&\geq& (e-1)x+1-e^x.
\end{eqnarray}
Let $z(x) = (e-1)x+1-e^x$. Since $z(0)=z(1)=0$, it is easy to see that $z(x)\geq 0$ when $0\leq x \leq 1$, which indicates that $y'(x)\geq 0$ when $0\leq x \leq 1$. Based on Eq.~(\ref{eq:dy}) and the fact that $e-1>1$, we have $y'(x)\geq 1$ when $x \geq 1$. As a result, $y'(x)\geq 0$ when $x \geq 0$, which means $y(x)\geq 0$. Therefore, 
\begin{eqnarray}
c(\Delta t) \leq e^{\frac{e}{2}x^2}-y(x) \leq e^{\frac{e}{2}x^2}.
\end{eqnarray}
\end{proof}

According to Lemma~\ref{lemma:cost_LOR}, the cost factor $e^{-iHt}$ has the upper bound 
\begin{eqnarray}
[c(t/N)]^N \leq e^{\frac{eh_{tot}^2t^2}{2N}}.
\end{eqnarray}
Therefore, the cost factor approaches one when the time step number $N$ increases. 

\subsection{Variance analysis}

In this section, first, we prove the general upper bound of the variance, then we work out the upper bound of the cost $c_k$. 

\subsubsection{Variance upper bound}

\begin{figure}[tbhp]
\centering
\includegraphics[width=\linewidth]{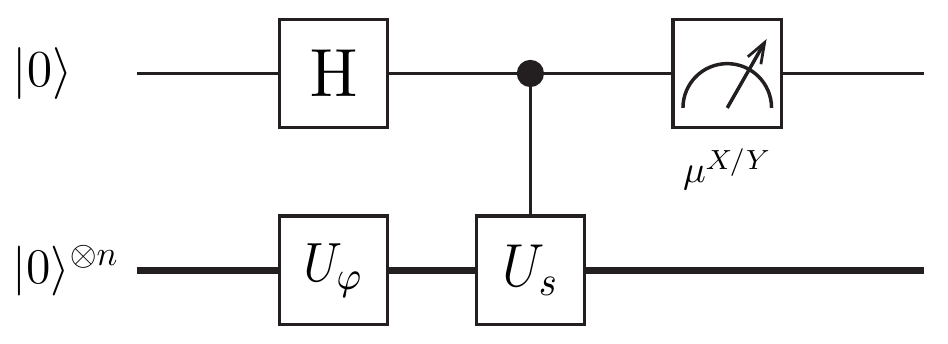}
\caption{
The Hadamard-test circuit $\mathcal{C}_s$. The qubit on the top is the ancilla qubit. The unitary operator $U_\varphi$ prepares the state $\ket{\varphi}$, i.e.~$U_\varphi\ket{0}^{\otimes n} = \ket{\varphi}$. When the ancilla qubit is measured in the $X$ (or $Y$) basis, the measurement outcome is an eigenvalue of the Pauli operator $X$ (or $Y$), i.e.~$\mu^X = \pm 1$ (or $\mu^Y = \pm 1$). 
}
\label{fig:HT}
\end{figure}

Given the LCU expression $A = \sum_s q_s U_s$, we can compute $\bra{\varphi}A\ket{\varphi}$ using the Monte Carlo method and the one-ancilla Hadamard-test circuit shown in Fig.~\ref{fig:HT}. In the Monte Carlo calculation, we sample the unitary operator $U_s$ with a probability of $\abs{q_s}/C_A$ in the spirit of importance sampling. The corresponding algorithm is given in Algorithm~\ref{alg:MC}. 

\begin{algorithm}[H]
\begin{algorithmic}[1]
\caption{Monte Carlo evaluation of an LCU expression of $A$.}
\label{alg:MC}
\State Input $\ket{\varphi}$, $\{(q_s,U_s)\}$ and $M$. 
\State $C_A \gets \sum_s \abs{q_s}$
\State $\hat{A} \gets 0$
\For{$l=1$ to $M$}
\State Choose $s$ with a probability of $\abs{q_s}/C_A$. 
\State Implement the circuit $\mathcal{C}_s$ for one shot, measure the ancilla qubit in the $X$ basis and record the measurement outcome $\mu^X$. 
\State Implement the circuit $\mathcal{C}_s$ for one shot, measure the ancilla qubit in the $Y$ basis and record the measurement outcome $\mu^Y$. 
\State $\hat{A} \gets \hat{A} + e^{i\arg(q_s)} (\mu^X+i\mu^Y)$
\EndFor
\State Output $\hat{A} \gets \frac{C_A}{M}\hat{A}$. 
\end{algorithmic}
\end{algorithm}

\begin{lemma}
According to Algorithm~\ref{alg:MC}, the estimator $\hat{A}$ is unbiased, and the variance upper bound in Eq.~(\ref{eq:variance_bound}) is true. 
\label{lem:variance}
\end{lemma}

\begin{proof}
First, we rewrite the LCU expression as 
\begin{eqnarray}
A = C_A \sum_s \frac{\abs{q_s}}{C_A} e^{i\arg(q_s)} U_s.
\end{eqnarray}
Then, 
\begin{eqnarray}
\bra{\varphi}A\ket{\varphi} = C_A \sum_s \frac{\abs{q_s}}{C_A} e^{i\arg(q_s)} \bra{\varphi}U_s\ket{\varphi}.
\label{eq:A_LCU}
\end{eqnarray}

Each unitary operator $U_s$ has a corresponding Hadamard-test circuit denoted by $\mathcal{C}_s$, which is shown in Fig.~\ref{fig:HT}. When the ancilla qubit is measured in the $W = X,Y$ basis, the measurement has a random outcome $\mu^W = \pm 1$. Let $P^W_{\mu^W}$ be the probability of the measurement outcome $\mu^W$ in $\mathcal{C}_s$. According to Ref.~\cite{ekert2002}, we have 
\begin{eqnarray}
\bra{\varphi}U_s\ket{\varphi} &=& \sum_{\mu^X=\pm 1}P^X_{\mu^X}\mu^X \notag\\
&& + i\sum_{\mu^Y=\pm 1}P^Y_{\mu^Y}\mu^Y.
\label{eq:A_mean}
\end{eqnarray}

The probability of $(s,\mu^X,\mu^Y)$ is $(\abs{q_s}/C_A) P^X_{\mu^X} P^Y_{\mu^Y}$. Using Eqs.~(\ref{eq:A_LCU}) and (\ref{eq:A_mean}), we have 
\begin{eqnarray}\label{eq:mean}
\bra{\varphi}A\ket{\varphi} &=& C_A \sum_{s,\mu^X,\mu^Y} \frac{\abs{q_s}}{C_A} P^X_{\mu^X} P^Y_{\mu^Y} e^{i\arg(q_s)} \notag\\
&&\times(\mu^X+i\mu^Y).
\end{eqnarray}
The corresponding Monte Carlo algorithm is given in Algorithm~\ref{alg:MC}. 

The estimator $\hat{A}$ is unbiased. According to Eq.~(\ref{eq:mean}), the expected value of $C_A e^{i\arg(q_s)} (\mu^X+i\mu^Y)$ is $\bra{\varphi}A\ket{\varphi}$. Therefore, the expected value of $\hat{A}$ is also $\bra{\varphi}A\ket{\varphi}$. Notice that $\hat{A}$ is the average of $C_A e^{i\arg(q_s)} (\mu^X+i\mu^Y)$ over $M$ samples. 

The variance of $\hat{A}$ is 
\begin{eqnarray}
&&\mathrm{Var}(\hat{A}) = \frac{1}{M} \left[ \sum_{s,\mu^X,\mu^Y} \frac{\abs{q_s}}{C_A} P^X_{\mu^X} P^Y_{\mu^Y} \right. \notag \\
&&\left. \times \absLR{C_A e^{i\arg(q_s)} (\mu^X+i\mu^Y)}^2 - \absLR{\bra{\varphi}A\ket{\varphi}}^2 \right] \notag \\
&& \leq \frac{2C_A^2}{M}.
\end{eqnarray}
\end{proof}
Here, we remark that the summation-form LCU expressions can be generalised to the integral-form LCU expressions, which inherit the same properties of the former.

\subsubsection{Cost upper bound}
We use a two-level LCU to construct $f_k$: First, $f_k$ is an integral of RTE operators $e^{-iHt}$; Second, RTE is realised using the leading-order-rotation method. Substituting Eq.~(\ref{eq:LOR2}) into Eq.~(\ref{eq:LCU}), we can obtain the eventual LCU expression of $f_k$, i.e.
\begin{align}
&f_k = \frac{i^{k-1}}{2^{\frac{k-1}{2}}\tau^{k-1}}
\int_{-\infty}^{+\infty}dt H_{k-1}\left(\frac{t}{\sqrt{2}\tau}\right)g_\tau(t) e^{i E_0 t} \notag\\
&\times \left[\sum_j \beta_j(t/N) e^{-i\sgn(h_j)\phi(t/N)\sigma_j} + T_0(t/N)\right]^N.
\end{align}
Substituting LCU expressions of $f_k$ and $f_q$ as well as the expression of $H$ into $A = f_k^\dag Hf_q,f_k^\dag f_q$, we can obtain the LCU expression of $A$. Then, $C_A = h_{tot}c_kc_q,c_kc_q$, respectively, where 
\begin{eqnarray}
c_k &=& \frac{1}{2^{\frac{k-1}{2}}\tau^{k-1}}
\int_{-\infty}^{+\infty}dt \absLR{H_{k-1}\left(\frac{t}{\sqrt{2}\tau}\right)}g_\tau(t) \notag\\
&&\times[c(t/N)]^N.
\label{eq:ck2}
\end{eqnarray}
Here, we have replaced $\left[\sum_r \absLR{v_r\left(\frac{t}{N}\right)}\right]^N$ with $[c(t/N)]^N$ in Eq.~(\ref{eq:ck}). 

To work out the cost of $A$, we have used that the cost factor is additive and multiplicative. Suppose LCU expressions of $A_1$ and $A_2$ are $A_1 = q_1U_1+q_1'U_1'$ and $A_2 = q_2U_2+q_2'U_2'$, respectively. The cost factors of $A_1$ and $A_2$ are $C_1 = \abs{q_1}+\abs{q_1'}$ and $C_2 = \abs{q_2}+\abs{q_2'}$, respectively. Substituting LCU expressions of $A_1$ and $A_2$ into $A = A_1+A_2$, the expression of $A$ reads $A = q_1U_1+q_1'U_1'+q_2U_2+q_2'U_2'$. Then, the cost factor of $A$ is $C_A = \abs{q_1}+\abs{q_1'}+\abs{q_2}+\abs{q_2'} = C_1+C_2$. Substituting LCU expressions of $A_1$ and $A_2$ into $A' = A_1A_2$, the expression of $A'$ reads $A' = (q_1U_1+q_1'U_1')(q_2U_2+q_2'U_2') = q_1q_2U_1U_2+q_1q_2'U_1U_2'+q_1'q_2U_1'U_2+q_1'q_2'U_1'U_2'$. Then, the cost factor of $A'$ is $C_A' = \abs{q_1q_2}+\abs{q_1q_2'}+\abs{q_1'q_2}+\abs{q_1'q_2'} = C_1C_2$. 

%In what follows, we give two upper bounds of $c_k$. 

The upper bound of $c_k$ is 
\begin{eqnarray}
c_k \leq c_k^{ub} &=& \frac{1}{2^{\frac{k-1}{2}}\tau^{k-1}}
\int_{-\infty}^{+\infty}dt \absLR{H_{k-1}\left(\frac{t}{\sqrt{2}\tau}\right)}g_\tau(t) \notag\\
&&\times e^{\chi\frac{t^2}{\tau^2}},
\end{eqnarray}
where $\chi=\frac{eh_{tot}^2\tau^2}{2N}$. Here, we have used Lemma~\ref{lemma:cost_LOR}. Notice that $c_k^{ub}$ increases monotonically with $\chi$ (i.e. decreases monotonically with $N$). Taking $\chi = \frac{1}{8}$ (i.e.~$N = 4eh_{tot}^2\tau^2$), we numerically evaluate the upper bound $c_k^{ub}$ and plot it in Fig.~\ref{fig:cost_upper_bound}. We can find that 
\begin{eqnarray}
c_k^{ub} \leq 2 \left(\frac{k-1}{e\tau^2}\right)^{\frac{k-1}{2}}.
\end{eqnarray}

\begin{figure}[tbhp]
\centering
\includegraphics[width=\linewidth]{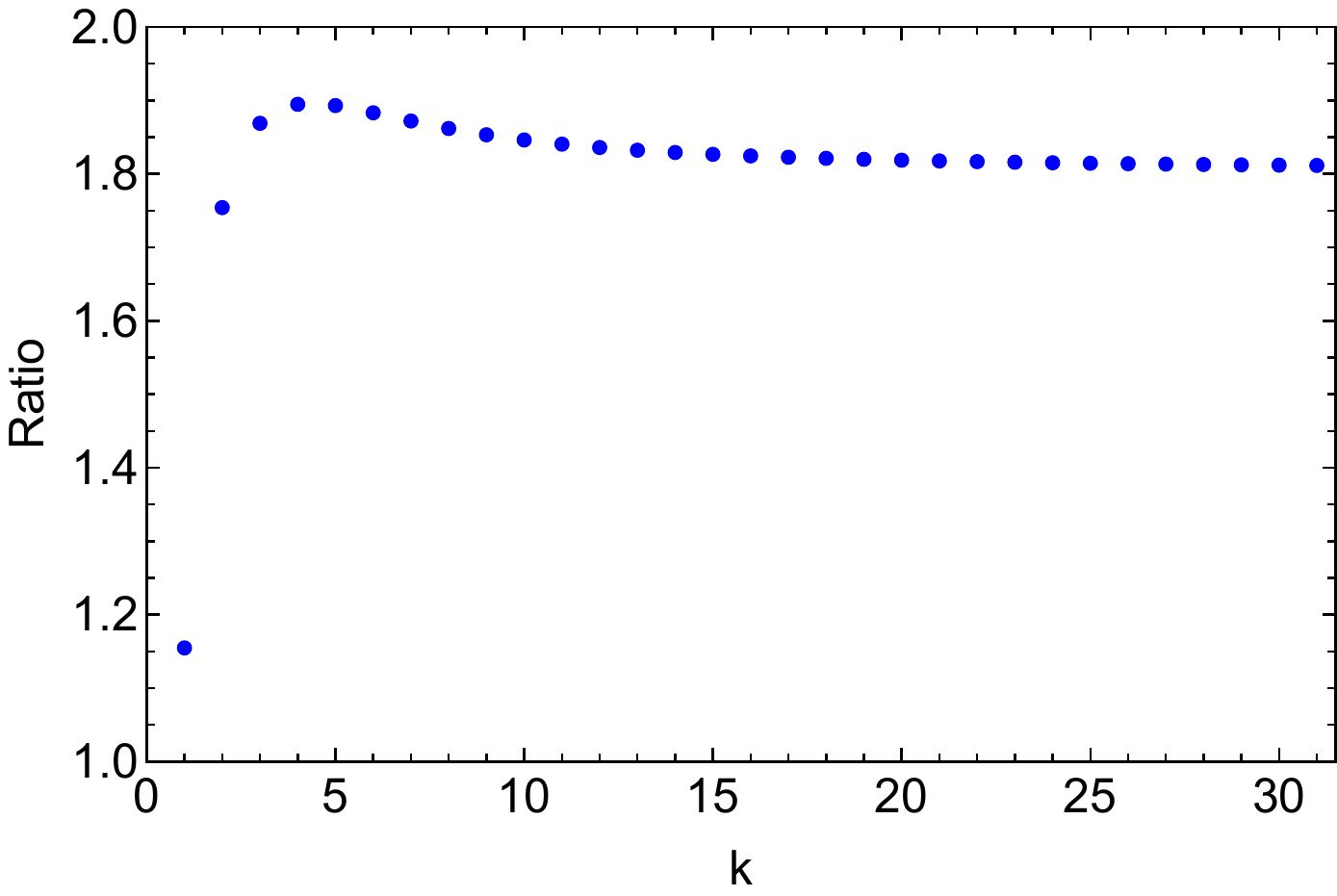}
\caption{The ratio between cost upper bound and norm upper bound, $c_k^{ub}/[(k-1)/(e\tau^2)]^{(k-1)/2}$.}
\label{fig:cost_upper_bound}
\end{figure}

%The second upper bound is obtained from Cram\'{e}r's inequality (\ref{eq:lemma_integral_4}). According to the inequality, we have 
%\begin{eqnarray}
%c_k^{ub} &\leq & \frac{1}{\tau\sqrt{2\pi}} \frac{\sqrt{(k-1)!}}{\tau^{k-1}}
%\int_{-\infty}^{+\infty}dt e^{-(\frac{1}{4}-\chi)\frac{t^2}{\tau^2}} \notag \\
%&=& \frac{1}{\tau^{k-1}}\sqrt{\frac{2(k-1)!}{1-4\chi}}.
%\end{eqnarray}
%This upper bound justifies taking \BLUE{$\chi = \frac{1}{8}$}. When $\tau$ is large, the upper bound decreases exponentially with $k$. 

\subsection{Pseudocode}

Given the LCU expression of $A$, we can compute $\bra{\varphi}A\ket{\varphi}$ according to Algorithm~\ref{alg:MC}. We would like to repeat how we compose the LCU expression of $A$: Substituting Eq.~(\ref{eq:LOR2}) into Eq.~(\ref{eq:LCU}), we can obtain the eventual LCU expression of $f_k$; Substituting LCU expressions of $f_k$ and $f_q$ as well as the expression of $H$ into $A = f_k^\dag Hf_q$ and $A = f_k^\dag f_q$ (corresponding to $\bfH_{k,q}$ and $\bfS_{k,q}$, respectively), we can obtain the LCU expression of $A$. Algorithm~\ref{alg:MC} does not involve details of the LCU expression. In this section, we give pseudocodes involving details. 

We remark that matrix entries of the rescaled basis $f'_k = f_k/c_k$ are $\hat{\bfH}'_{k,q} = \hat{\bfH}_{k,q}/(c_kc_q)$ and $\hat{\bfS}'_{k,q} = \hat{\bfS}_{k,q}/(c_kc_q)$, where $\hat{\bfH}_{k,q}$ and $\hat{\bfS}_{k,q}$ are computed according to the pseudocodes. 

In pseudocodes, we use notations $\mathbf h = (h_1,h_2,\ldots)$ and $\boldsymbol \sigma = (\sigma_1,\sigma_2,\ldots)$ to represent the Hamiltonian $H = \sum_j h_j\sigma_j$. We define probability density functions 
\begin{eqnarray}
p_{\tau,k}(t) &=& \frac{1}{c_k 2^{\frac{k-1}{2}}\tau^{k-1}}
\absLR{H_{k-1}\left(\frac{t}{\sqrt{2}\tau}\right)}g_\tau(t) \notag\\
&&\times [c(t/N)]^N,
\label{eq:PDF}
\end{eqnarray}
and probabilities 
\begin{eqnarray}
P_\mathrm{L}(\Delta t) &=& \frac{\sqrt{1+h_{tot}^2\Delta t^2}}{c(\Delta t)}, \\
P_\mathrm{T}(\Delta t) &=& \frac{e^{h_{tot}\abs{\Delta t}} - (1+h_{tot}\abs{\Delta t})}{c(\Delta t)}.
\end{eqnarray}
We use $\text{P{\scriptsize OISSON}}(\lambda)$ to denote a subroutine that returns $k$ with a probability of $\lambda^ke^{-\lambda}/k!$. The LCU expression of $f_k$ is sampled according to Algorithms~\ref{alg:basis}, \ref{alg:RTE} and \ref{alg:LOR}. In Algorithm~\ref{alg:basis}, we generate random time $t$ according to the distribution $p_{\tau,k}(t)$. In Algorithm~\ref{alg:RTE}, we realise the RTE $e^{-iHt}$ by $N$ steps leading-order-rotation. In Algorithm~\ref{alg:LOR}, we randomly sample rotations or Pauli operators according to Eq.~(\ref{eq:LOR2}).

\begin{algorithm}[H]
\begin{algorithmic}[1]
\caption{Measurement of matrix entry $\bfH_{k,q}$.}
\label{alg:Hkq}
\State Input $\ket{\varphi}$, $(\mathbf h,\boldsymbol \sigma)$, $E_0$, $\tau$, $N$ and $M$. 
\State $h_{tot} \gets \sum_j \abs{h_j}$
\State Compute $c_k$ and $c_q$. 
\State $\hat{\bfH}_{k,q} \gets 0$
\For{$l=1$ to $M$}
\State Choose $j$ with a probability of $\abs{h_j}/h_{tot}$. 
\State $(\bar{v}_k,\bar{V}_k)$ $\gets$ \Call{BasisGen}{$\mathbf h$,$\boldsymbol \sigma$,$E_0$,$\tau$,$N$,$k$}
\State $(\bar{v}_q,\bar{V}_q)$ $\gets$ \Call{BasisGen}{$\mathbf h$,$\boldsymbol \sigma$,$E_0$,$\tau$,$N$,$q$}
\State Implement the circuit $\mathcal{C}_s$ with $U_s = \bar{V}_k^\dag \sigma_j \bar{V}_q$ for one shot, measure the ancilla qubit in the $X$ basis and record the measurement outcome $\mu^X$. 
\State Implement the circuit $\mathcal{C}_s$ with $U_s = \bar{V}_k^\dag \sigma_j \bar{V}_q$ for one shot, measure the ancilla qubit in the $Y$ basis and record the measurement outcome $\mu^Y$. 
\State $\hat{\bfH}_{k,q} \gets \hat{\bfH}_{k,q} + e^{i\arg(\bar{v}_k^* h_j \bar{v}_q)} (\mu^X+i\mu^Y)$
\EndFor
\State Output $\hat{\bfH}_{k,q} \gets \frac{h_{tot}c_kc_q}{M}\hat{\bfH}_{k,q}$. 
\end{algorithmic}
\end{algorithm}

\begin{algorithm}[H]
\begin{algorithmic}[1]
\caption{Measurement of matrix entry $\bfS_{k,q}$.}
\label{alg:Skq}
\State Input $\ket{\varphi}$, $(\mathbf h,\boldsymbol \sigma)$, $E_0$, $\tau$, $N$ and $M$. 
\State Compute $c_k$ and $c_q$. 
\State $\hat{\bfS}_{k,q} \gets 0$
\For{$l=1$ to $M$}
\State $(\bar{v}_k,\bar{V}_k)$ $\gets$ \Call{BasisGen}{$\mathbf h$,$\boldsymbol \sigma$,$E_0$,$\tau$,$N$,$k$}
\State $(\bar{v}_q,\bar{V}_q)$ $\gets$ \Call{BasisGen}{$\mathbf h$,$\boldsymbol \sigma$,$E_0$,$\tau$,$N$,$q$}
\State Implement the circuit $\mathcal{C}_s$ with $U_s = \bar{V}_k^\dag \bar{V}_q$ for one shot, measure the ancilla qubit in the $X$ basis and record the measurement outcome $\mu^X$. 
\State Implement the circuit $\mathcal{C}_s$ with $U_s = \bar{V}_k^\dag \bar{V}_q$ for one shot, measure the ancilla qubit in the $Y$ basis and record the measurement outcome $\mu^Y$. 
\State $\hat{\bfS}_{k,q} \gets \hat{\bfS}_{k,q} + e^{i\arg(\bar{v}_k^* \bar{v}_q)} (\mu^X+i\mu^Y)$
\EndFor
\State Output $\hat{\bfS}_{k,q} \gets \frac{c_kc_q}{M}\hat{\bfS}_{k,q}$. 
\end{algorithmic}
\end{algorithm}

\begin{algorithm}[H]
\begin{algorithmic}[1]
\caption{Basis generator $f_k$.}
\label{alg:basis}
\Function{BasisGen}{$\mathbf h$,$\boldsymbol \sigma$,$E_0$,$\tau$,$N$,$k$}
\State Generate random time $t$ according to the distribution $p_{\tau,k}(t)$. 
\State $(\bar{v},\bar{V})$ $\gets$ \Call{RTE}{$\mathbf h$,$\boldsymbol \sigma$,$N$,$t$}
\State $\bar{v} \gets \bar{v} \times \frac{i^{k-1}}{2^{\frac{k-1}{2}}\tau^{k-1}} H_{k-1}\left(\frac{t}{\sqrt{2}\tau}\right)g_\tau(t)e^{iE_0t}$
\State Output $(\bar{v},\bar{V})$. 
\EndFunction
\end{algorithmic}
\end{algorithm}

\begin{algorithm}[H]
\begin{algorithmic}[1]
\caption{Real-time evolution.}
\label{alg:RTE}
\Function{RTE}{$\mathbf h$,$\boldsymbol \sigma$,$N$,$t$}
\State $\bar{v} \gets 1$
\State $\bar{V} \gets \openone$
\For{$i=1$ to $N$}
\State $(v,V)$ $\gets$ \Call{LORF}{$\mathbf h$,$\boldsymbol \sigma$,$t/N$}
\State $\bar{v} \gets \bar{v}\times v$
\State $\bar{V} \gets \bar{V}\times V$
\EndFor
\State Output $(\bar{v},\bar{V})$. 
\EndFunction
\end{algorithmic}
\end{algorithm}

\begin{algorithm}[H]
\begin{algorithmic}[1]
\caption{Leading-order-rotation formula.}
\label{alg:LOR}
\Function{LORF}{$\mathbf h$,$\boldsymbol \sigma$,$\Delta t$}
\State Choose $\mathrm{O}$ from $\mathrm{L}$ and $\mathrm{T}$ with probabilities of $P_\mathrm{L}(\Delta t)$ and $P_\mathrm{T}(\Delta t)$, respectively. 
\If{$\mathrm{O} = \mathrm{L}$}
\State Choose $j$ with a probability of $\abs{h_j}/h_{tot}$. 
\State $v \gets \beta_j(\Delta t)$
\State $V \gets e^{-i\sgn(h_j)\phi(\Delta t)\sigma_j}$
\ElsIf{$\mathrm{O} = \mathrm{T}$}
\While{$k<2$}
\State $k \gets \Call{Poisson}{h_{tot}\abs{\Delta t}}$
\EndWhile
\State $v \gets 1/k!$
\State $V \gets \openone$
\For{$a=1$ to $k$}
\State Choose $j$ with a probability of $\abs{h_j}/h_{tot}$. 
\State $v \gets v\times (-ih_j\Delta t)$
\State $V \gets V\times \sigma_j$
\EndFor
\EndIf
\State Output $(v,V)$. 
\EndFunction
\end{algorithmic}
\end{algorithm}

\section{Details in numerical calculation}
\label{app:details}

In this section, first, we describe models and reference states taken in the benchmarking, and then we give details about instances and parameters. 

\subsection{Models and reference states}

Two models are used in the benchmarking: the anti-ferromagnetic Heisenberg model 
\begin{eqnarray}
H = J\sum_{\langle i,j\rangle}(\sigma^X_i\sigma^X_j+\sigma^Y_i\sigma^Y_j+\sigma^Z_i\sigma^Z_j),
\end{eqnarray}
where $\sigma^X_i$, $\sigma^Y_i$ and $\sigma^Z_i$ are Pauli operators of the $i$th qubit; and the Hubbard model 
\begin{eqnarray}
H &=& -J\sum_{\langle i,j\rangle}\sum_{\sigma=\uparrow,\downarrow}(a_{i,\sigma}^\dag a_{j,\sigma}+a_{j,\sigma}^\dag a_{i,\sigma}) \notag \\
&+& U\sum_i \left(a_{i,\uparrow}^\dag a_{i,\uparrow}-\frac{\openone}{2}\right)\left(a_{i,\downarrow}^\dag a_{i,\downarrow}-\frac{\openone}{2}\right),
\end{eqnarray}
where $a_{i,\sigma}$ is the annihilation operator of the $i$th orbital and spin $\sigma$. For the Heisenberg model, we take the spin number as ten, and for the Hubbard model, we take the orbital number as five. Then both models can be encoded into ten qubits. For the Hubbard model, we take $U = J$. Without loss of generality, we normalise the Hamiltonian for simplicity: We take $J$ such that $\norm{H}_2 = 1$, i.e.~eigenvalues of $H$ are in the interval $[-1,1]$. 

For each model, we consider three types of lattices: chain, ladder and randomly generated graphs, see Fig.~\ref{fig:lattices}. We generate a random graph in the following way: i) For a vertex $v$, randomly choose a vertex $u\neq v$ and add the edge $(v,u)$; ii) Repeat step-i two times for the vertex $v$; iii) Implement steps i and ii for each vertex $v$. On a graph generated in this way, each vertex is connected to four vertices on average. Notice that some pairs of vertices are connected by multiple edges, and then the interaction between vertices in the pair is amplified by the number of edges. 

\begin{figure}[tbhp]
\centering
\includegraphics[width=\linewidth]{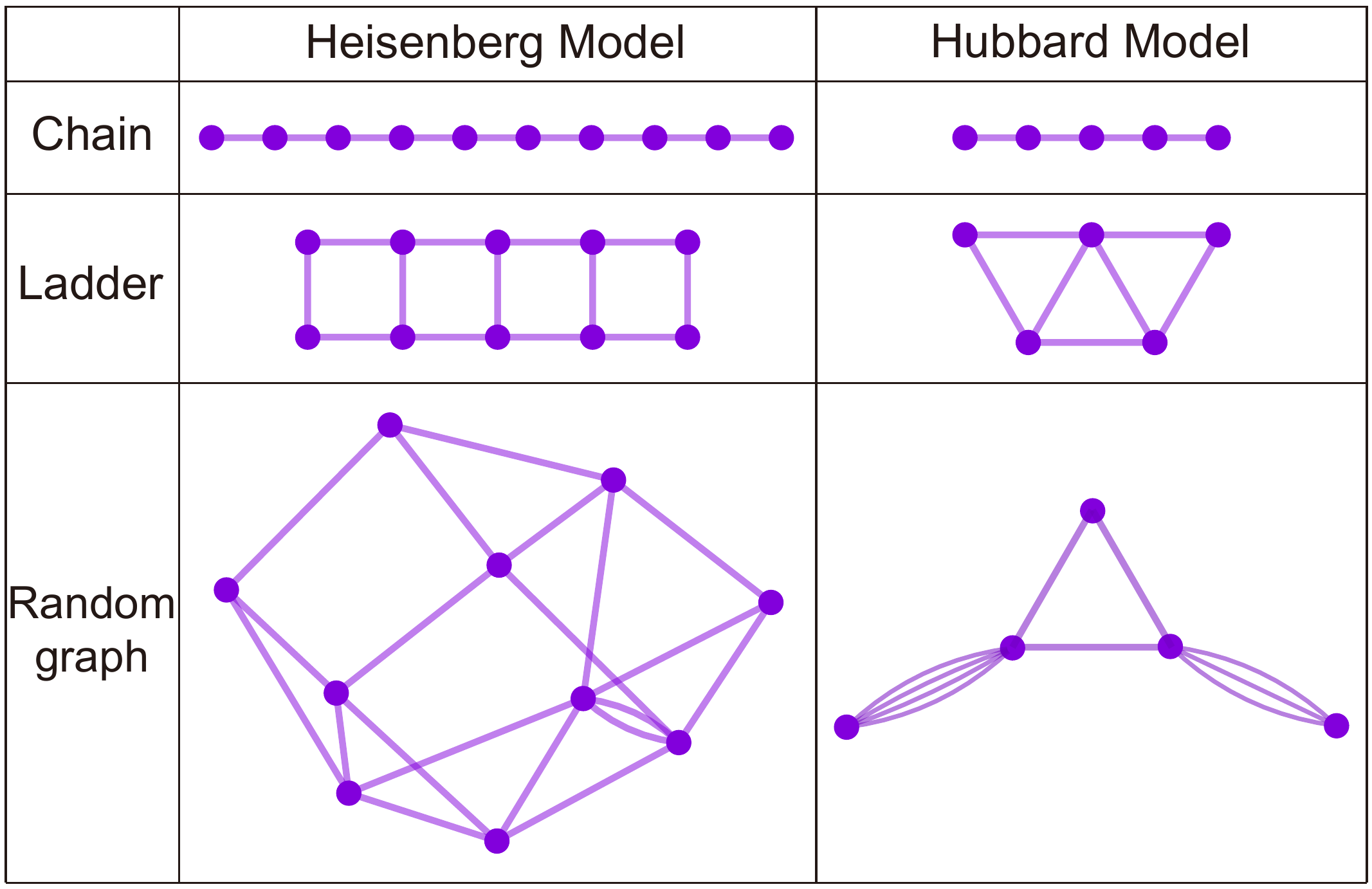}
\caption{Lattices of the Heisenberg model and Hubbard model. One of the randomly generated graphs for each model is illustrated in the figure as an example.}
\label{fig:lattices}
\end{figure}

It is non-trivial to choose and prepare a good reference state that has a sufficiently large overlap with the true ground state. Finding the ground-state energy of a local Hamiltonian is QMA-hard~\cite{kitaev2002,kempe2006}. If one can prepare a good reference state, then quantum computing is capable of solving the ground-state energy using quantum phase estimation~\cite{gharibian2022}. So far, there does not exist a universal method for reference state preparation; see Ref.~\cite{lee2022} for relevant discussions. Despite this, there are some practical ways of choosing and preparing the reference state. For fermion systems, we can take the mean-field approximate solution, i.e.~a Hartree-Fock state, as the reference state. For general systems, one may try adiabatic state preparation, etc. We stress that preparing good reference states is beyond the scope of this work. Here, we choose two reference states which are likely to overlap with true ground states as examples. For the Heisenberg model, we take the pairwise singlet state 
\begin{eqnarray}
\ket{\varphi} &=& \ket{\Phi}_{1,2}\otimes \ket{\Phi}_{3,4}\otimes \cdots,
\label{eq:varphi}
\end{eqnarray}
where 
\begin{eqnarray}
\ket{\Phi}_{i,j} = \frac{1}{\sqrt{2}}\left(\ket{0}_i\otimes\ket{1}_j-\ket{1}_i\otimes\ket{0}_j\right)
\end{eqnarray}
is the singlet state of spins $i$ and $j$. For the Hubbard model, we take a Hartree-Fock state as the reference state: We compute the ground state of the one-particle Hamiltonian (which is equivalent to taking $U=0$) and take the ground state of the one-particle Hamiltonian (a Slater determinant) as the reference state. 

\subsection{Instances}

We test the algorithms listed in Table~\ref{table} with many instances. Each instance is a triple $(model,lattice,d)$, in which $model$ takes one of the two models (Heisenberg model and Hubbard model), $lattice$ takes a chain, ladder or random graph, and $d$ is the dimension of the subspace. 

Instances for computing empirical distributions in Fig.~\ref{fig:distribution} consist of six groups: the Heisenberg model on the chain, ladder and random graphs, and the Hubbard model on the chain, ladder and random graphs. For chain and ladder, we evaluate each subspace dimension $d = 2,3,\ldots,30$; For a random graph, we only evaluate one subspace dimension randomly taken in the interval. Some instances are discarded. To avoid any potential issue caused by the precision of floating-point numbers, we discard instances that the subspace error $\epsilon_K$ of the P algorithm is lower than $10^{-9}$ due to a large $d$. We also discard instances that $\epsilon_K$ of the P algorithm is higher than $10^{-2}$ due to a small $d$, because the dimension may be too small to achieve an interesting accuracy in this case. For randomly generated graphs, the two reference states may have a small overlap with the true ground state. Because a good reference state is necessary for all quantum KSD algorithms, we discard graphs with $p_g<10^{-3}$. Eventually, we have eight instances of the Heisenberg chain, seven instances of the Heisenberg Ladder, twenty-two instances of the Hubbard chain and twenty instances of the Hubbard ladder. For each model, we generate a hundred random-graph instances. The total number of instances is $233$. 

\subsection{Parameters}

\begin{table*}[tbhp]
\centering
\begin{tabular}{c|c|c|c|c|c|c} 
\hline\hline
Abbr. & $E_0$ & $\tau$ & $\Delta t$ & $\Delta E$ & $C_{\bfH}$ & $C_{\bfS}$ \\ 
\hline
P & $E_g+1$ & N/A & N/A & N/A & $1$ & $1$ \\
CP & $0$ & N/A & N/A & N/A & $1$ & $1$ \\
GP & $[E_g-0.1,E_g+0.1]$ & Solving Eq.~(\ref{eq:tau1}) & N/A & N/A & $h_{tot}\geq 1$ & $1$ \\
IP & $E_g-1$ & N/A & N/A & N/A & $1$ & $1$ \\
ITE & $E_g$ & Solving Eq.~(\ref{eq:tau2}) & N/A & N/A & $1$ & $1$ \\
RTE & $E_g$ & N/A & Minimising $\epsilon_K$ & N/A & $1$ & $1$ \\
F & $E_g$ & Solving Eq.~(\ref{eq:tau1}) & $2\tau/(L-1)$ & Minimising $\epsilon_K$ & $1$ & $1$ \\
\hline\hline
\end{tabular}
\caption{Parameters taken in the numerical calculation.}
\label{parameters}
\end{table*}

In this section, we give parameters taken in each algorithm listed in Table~\ref{table}, namely, $E_0$, $\tau$, $\Delta t$, $\Delta E$, $C_{\bfH}$ and $C_{\bfS}$. We summarise these parameters in Table~\ref{parameters}. 

For $E_0$, we expect that taking $E_0$ close to $E_g$ is optimal for our GP algorithm. As the exact value of $E_g$ is unknown, we assume that we have a preliminary estimation of the ground-state energy with uncertainty as large as $10\%$ of the entire spectrum, i.e.~we take $E_0 \in [E_g-0.1,E_g+0.1]$. For other algorithms, we take $E_0$ by assuming that the exact value of $E_g$ is known. In the P algorithm, we take $E_0 = E_g+1$, such that the ground state is the eigenstate of $H-E_0$ with the largest absolute eigenvalue and $\norm{H-E_0}_2 = 1$. Similarly, in the IP algorithm, we take $E_0 = E_g-1$, such that the ground state is the eigenstate of $(H-E_0)^{-1}$ with the largest absolute eigenvalue and $\norm{(H-E_0)^{-1}}_2 = 1$. In the ITE algorithm, $E_0$ causes a constant factor $e^{\tau(k-1)E_0}$ in each operator $f_k$, i.e.~determines the spectral norm of $f_k$. Because the variance is related to the norm, a large $E_0$ is problematic. Therefore, in the ITE algorithm, we take $E_0 = E_g$, such that $\norm{f_k}_2 = 1$ for all $k$. In the F algorithm, we expect that $E_0 = E_g$ is optimal because $f_1$ is an energy filter centred at the ground-state energy in this case. Therefore, we take $E_0 = E_g$ in the F algorithm. The RTE algorithm is closely related to the F algorithm. Therefore, we also take $E_0 = E_g$ in the RTE algorithm. 

We remark that in the GP algorithm, we take random $E_0$ uniformly distributed in the interval $[E_g-0.1,E_g+0.1]$ to generate data in Fig.~\ref{fig:distribution}, and we take $E_0 = -E_g+i\delta E$, where $i = -50,-9,\ldots,-1,0,1,\ldots,9,50$ and $\delta E = 0.002$, to generate data in Fig.~\ref{fig:error_vs_overhead}. We remark that we have normalised the Hamiltonian such that $\norm{H}_2 = 1$. 

For $\Delta t$ in the F algorithm, we take $\Delta t = 2\tau/(L-1)$. For simplicity, we take the limit $L\rightarrow +\infty$, i.e. 
\begin{eqnarray}
f_k &=& \frac{1}{2\tau}\int_{-\tau}^{\tau} dt e^{-i\left[H-E_0-\Delta E\left(k-1\right)\right]t} \notag \\
&=& \frac{\sin \left[H-E_0-\Delta E\left(k-1\right)\right]\tau}{\left[H-E_0-\Delta E\left(k-1\right)\right]\tau}.
\end{eqnarray}
Notice that this $f_k$ is an energy filter centred at $E_0+\Delta E\left(k-1\right)$, and the filter is narrower when $\tau$ is larger. 

Now, we have three algorithms that have the parameter $\tau$: GP, ITE and F. Similar to the F algorithm, $f_1$ in the GP algorithm is an energy filter centred at $E_0$. For these two algorithms, if $E_0 = E_g$, $f_1\ket{\varphi}$ converges to the true ground state in the limit $\tau\rightarrow +\infty$. It is also similar in the ITE algorithm, in which $f_d$ is a projector onto the ground state, and $f_d\ket{\varphi}$ converges to the true ground state in the limit $\tau\rightarrow +\infty$. Realising a large $\tau$ is at a certain cost, specifically, the circuit depth increases with $\tau$~\cite{berry2007,motta2019}. Therefore, it is reasonable to take a finite $\tau$. Next, we give the protocol for determining the value of $\tau$ in each algorithm. 

Without getting into details about realising the filters and projectors, we choose $\tau$ under the assumption that if filters and projectors have the same power in computing the ground-state energy, they probably require similar circuit depths. In GP and F algorithms, if $E_0 = E_g$, the energy error achieved by filters reads $\bfH_{1,1}/\bfS_{1,1}-E_g$; and in the ITE algorithm, the energy error achieved by the projector reads $\bfH_{d,d}/\bfS_{d,d}-E_g$. We take $\tau$ such that errors achieved by filters and projectors take the same value $\epsilon_B$. Specifically, for the GP and F algorithms, we determine the value of $\tau$ by solving the equation (taking $E_0 = E_g$) 
\begin{eqnarray}
\frac{\bfH_{1,1}(\tau)}{\bfS_{1,1}(\tau)} - E_g = \epsilon_B;
\label{eq:tau1}
\end{eqnarray}
and for the ITE algorithm, we determine the value of $\tau$ by solving the equation 
\begin{eqnarray}
\frac{\bfH_{d,d}(\tau)}{\bfS_{d,d}(\tau)} - E_g = \epsilon_B.
\label{eq:tau2}
\end{eqnarray}
To choose the value of $\epsilon_B$, we take the P algorithm as the standard, in which $f_d$ is a projector onto the ground state, i.e.~$f_d\ket{\varphi}$ converges to the true ground state in the limit $d\rightarrow +\infty$. Overall, we determine the value of $\tau$ in the following way: Given an instance $(model,lattice,d)$, i) first, compute the energy error achieved by the projector in the P algorithm, take $\epsilon_B = \bfH_{d,d}/\bfS_{d,d}-E_g$; then, ii) solve equations of $\tau$. In this way, filters and projectors in P, GP, ITE and F algorithms have the same power. 

For $\Delta t$ in the RTE algorithm and $\Delta E$ in the F algorithm, there are works on how to choose them in the literature~\cite{klymko2022,shen2022,cortes2022}. In this work, we simply determine their values by a grid search. For the RTE algorithm, we take $\Delta t = i\delta t$, where $i = 1,2,\ldots,100$ and $\delta t = \frac{2\pi}{100}$; we compute the subspace error $\epsilon_K$ for all $\Delta t$; and we choose $\Delta t$ of the minimum $\epsilon_K$. For the F algorithm, we take $\Delta E = i\delta E$, where $i = 1,2,\ldots,100$ and $\delta E = \frac{2}{100d}$ (In this way, when we take the largest $\Delta E$, filters span the entire spectrum); we compute the subspace error $\epsilon_K$ for all $\Delta E$; and we choose $\Delta E$ of the minimum $\epsilon_K$. 

For $C_{\bfH}$ and $C_{\bfS}$, we take $C_{\bfH} = h_{tot}$ and $C_{\bfS} = 1$ in our GP algorithm following the analysis in Appendix~\ref{app:algorithm}. For other algorithms, we take these two parameters in the following way. For P, IP, ITE and RTE algorithms, we take $C_{\bfH}$ and $C_{\bfS}$ according to spectral norms (the lower bound of cost) without getting into details about measuring matrix entries. In these four algorithms, $\norm{f_k^\dag Hf_q}_2 = \norm{f_k^\dag f_q}_2 = 1$ for all $k$ and $q$, therefore, we take $C_{\bfH} = C_{\bfS} = 1$. In the CP algorithm, we take $C_{\bfH}$ and $C_{\bfS}$ in a similar way. Because $\norm{H/h_{tot}}_2 \leq 1$, the spectrum of $H/h_{tot}$ is in the interval $[-1,1]$. In this interval, the Chebyshev polynomial of the first kind takes values in the interval $[-1,1]$. Therefore, $\norm{T_k(H/h_{tot})}_2 \leq 1$, and consequently $\norm{f_k^\dag Hf_q}_2, \norm{f_k^\dag f_q}_2 \leq 1$. These norms depend on the spectrum of $H$. For simplicity, we take the upper bound of norms, i.e.~$C_{\bfH} = C_{\bfS} = 1$, in the CP algorithm. In the F algorithm, each $f_k$ is a linear combination of $e^{-iHt}$ operators, i.e.~$f_k$ is expressed in the LCU form, and the cost factor of the LCU expression is one. Therefore, we take $C_{\bfH} = C_{\bfS} = 1$ in the F algorithm, assuming that $H$ is expressed in the LCU form with a cost factor of one (Notice that one is the lower bound). 

\section{Derivation of Eq.~(\ref{eq:gamma_S})}
\label{app:gamma}
Define the function $f(\eta)=E'(\eta,\bfa)=\frac{\bfa^\dag \bfH \bfa+2C_{\bfH}\eta\bfa^\dag\bfa}{\bfa^\dag \bfS \bfa+2C_{\bfS}\eta\bfa^\dag\bfa}$. Consider a small $\eta$ and the Taylor expansion of $f(\eta)$. We have $f(0)=\frac{\bfa^\dag \bfH \bfa}{\bfa^\dag \bfS \bfa}$ and 
\begin{eqnarray}
	f'(\eta)&=&\frac{2\bfa^\dag\bfa(C_{\bfH}\bfa^\dag \bfS \bfa-C_{\bfS} \bfa^\dag \bfH \bfa)}{(\bfa^\dag \bfS \bfa+2C_{\bfS}\eta\bfa^\dag\bfa)^2},\notag\\
	f'(0)&=&\frac{2\bfa^\dag\bfa(C_{\bfH}\bfa^\dag \bfS \bfa-C_{\bfS} \bfa^\dag \bfH \bfa)}{(\bfa^\dag \bfS \bfa)^2}\notag\\
	&=&\frac{\bfa^\dag\bfa}{\bfa^\dag \bfS \bfa}2\left(C_{\bfH}-C_{\bfS} \frac{\bfa^\dag \bfH \bfa}{\bfa^\dag \bfS \bfa}\right).
\end{eqnarray}
The minimum value of the zeroth-order term $\frac{\bfa^\dag \bfH \bfa}{\bfa^\dag \bfS \bfa}$ is $E_{min}$ cccording to the Rayleigh quotient theorem~\cite{horn2012}. Take
this minimum value, we have $E' \simeq E_{min} + s\eta$, where $s= f'(0) \propto\frac{\bfa^\dag \bfa}{\bfa^\dag \bfS \bfa}$ is the first derivative. The solution to $E' = E_g+\epsilon$ is $\eta \simeq (E_g+\epsilon-E_{min})/s$, then
\begin{eqnarray}
	\gamma &\simeq& \frac{\epsilon^2}{16\norm{H}_2^2}\frac{4\left(C_{\bfH}-C_{\bfS} \frac{\bfa^\dag \bfH \bfa}{\bfa^\dag \bfS \bfa}\right)^2}{(E_g+\epsilon-E_{min})^2}\left(\frac{p_g\bfa^\dag \bfa}{\bfa^\dag \bfS \bfa}\right)^2\notag\\
 &=&u^2\left(\frac{p_g\bfa^\dag \bfa}{\bfa^\dag \bfS \bfa}\right)^2,
\end{eqnarray}
where
\begin{eqnarray}
	u&=&\frac{\epsilon \absLR{C_{\bfH}-C_{\bfS}\frac{\bfa^\dag \bfH \bfa}{\bfa^\dag \bfS \bfa}}}{2\norm{H}_2\absLR{E_g+\epsilon-E_{min}}}\notag\\
	&\leq& \frac{(C_{\bfH}+\norm{H}_2C_{\bfS})\epsilon}{2\norm{H}_2(\epsilon-\epsilon_K)}\approx C_{\bfS}
\end{eqnarray}
under assumptions $C_{\bfH}\approx\norm{H}_2C_{\bfS}$ and $\epsilon-\epsilon_K\approx\epsilon$.

\section{Composing Chebyshev polynomials}
\label{app:CPP}

\subsection{Chebyshev-polynomial projector}

In this section, we explain how to use Chebyshev polynomials as projectors onto the ground state. 

The $n$th Chebyshev polynomial of the first kind is 
\begin{align}
T_n(z) = \left\{\begin{array}{cc}
\cos(n\arccos z), & \text{if }\abs{z}\leq 1, \\
\sgn(z)^n\cosh(n\,\mathrm{arccosh}\abs{z}), & \text{if }\abs{z}>1.
\end{array}
\right.
\end{align}
Chebyshev polynomials have the following properties: i) When $\abs{z}\leq 1$, $\abs{T_n(z)}\leq 1$; and ii) when $\abs{z}>1$, $\abs{T_n(z)}>(\abs{z}^n+\abs{z}^{-n})/2$ increases exponentially with $n$. 

%\BLUE{$$T_n(z) = f(w) = \frac{1}{2}(w^n+w^{-n})$$
%$$w = z+\sqrt{z^2-1}>z$$
%$$\frac{df}{dw} = \frac{n}{2w}(w^n-w^{-n})>0$$
%$$f(w)>f(z)$$}

Let's consider the spectral decomposition of the Hamiltonian, $H = \sum_{m=1}^{d_{\mathcal{H}}} E_m\ketbra{\psi_m}{\psi_m}$. Here, $d_{\mathcal{H}}$ is the dimension of the Hilbert space, $E_m$ are eigenenergies of the Hamiltonian, and $\ket{\psi_m}$ are eigenstates. We suppose that eigenenergies are sorted in ascending order, i.e.~$E_i\leq E_j$ if $i<j$. Then, $E_1 = E_g$ is the ground-state energy (accordingly, $\ket{\psi_1} = \ket{\psi_g}$ is the ground state) and $E_2$ is the energy of the first excited state. The energy gap between the ground state and the first excited state is $\Delta = E_2-E_1$. 

To compose a projector, we take 
\begin{eqnarray}
Z = 1-\frac{H-E_2}{\norm{H}_2}.
\end{eqnarray}
Notice that $Z$ and $H$ have the same eigenstates. The spectral decomposition of $Z$ is $Z= \sum_{m=1}^{d_{\mathcal{H}}} z_m\ketbra{\psi_m}{\psi_m}$, where 
\begin{eqnarray}
z_m = 1-\frac{E_m-E_2}{\norm{H}_2}.
\end{eqnarray}
Then, the ground state corresponds to $z_1 = 1+\frac{\Delta}{\norm{H}_2}>1$ (suppose the gap is finite), and the first excited state corresponds to $z_2 = 1$. Because $\abs{E_m-E_2}\leq 2\norm{H}_2$, $z_m\geq -1$ for all $m$. Therefore, except the ground state, $z_m$ of all excited states (i.e.~$m\geq 2$) is in the interval $[-1,1]$. 

The projector reads 
\begin{eqnarray}
\frac{T_n(Z)}{T_n(z_1)} &=& \sum_{m=1}^{d_{\mathcal{H}}} \frac{T_n(z_m)}{T_n(z_1)} \ketbra{\psi_m}{\psi_m} \notag \\
&=& \ketbra{\psi_g}{\psi_g} + \Omega,
\end{eqnarray}
where 
\begin{eqnarray}
\Omega &=& \sum_{m=2}^{d_{\mathcal{H}}} \frac{T_n(z_m)}{T_n(z_1)} \ketbra{\psi_m}{\psi_m}.
\end{eqnarray}
$\Omega$ and $H$ have the same eigenstates. Because $\abs{T_n(z_m)}\leq 1$ when $m\geq 2$, 
\begin{eqnarray}
\norm{\Omega}_2 \leq \frac{1}{T_n(z_1)}.
\end{eqnarray}

The key in using Chebyshev polynomials as projectors is laying all excited states in the interval $z\in [-1,1]$ and leaving the ground state out of the interval. For the CP basis, all eigenstates are in the interval $z\in [-1,1]$. Therefore, $f_k$ operators of the CP basis are not projectors in general. 

\subsection{Expanding the projector as Hamiltonian powers}

In this section, we expand the Chebyshev-polynomial projector as a linear combination of Hamiltonian powers $(H-E_0)^{k-1}$. We focus on the case that $E_0 = E_g$. 

The explicit expression of $T_n(z)$ ($n>0$) is 
\begin{eqnarray}
T_n(z) &=& n\sum_{m=0}^n (-2)^m\frac{(n+m-1)!}{(n-m)!(2m)!} \notag\\
&& \times (1-z)^m.
\end{eqnarray}
Because 
\begin{eqnarray}
(1-Z)^m &=& \sum_{l=0}^m \frac{m!}{(m-l)!l!} \notag\\
&&\times \frac{(E_0-E_2)^{m-l}(H-E_0)^l}{\norm{H}_2^m},
\end{eqnarray}
we have 
\begin{eqnarray}
T_n(Z) = \sum_{l=0}^n b_l \frac{(H-E_0)^l}{\norm{H}_2^l},
\end{eqnarray}
where 
\begin{eqnarray}
b_l &=& n\sum_{m=l}^n (-2)^m\frac{(n+m-1)!}{(n-m)!(2m)!} \frac{m!}{(m-l)!l!}\notag\\
&&\times \frac{(E_0-E_2)^{m-l}}{\norm{H}_2^{m-l}}.
\end{eqnarray}

Now, we give an upper bound of $\abs{b_l}$. First, we consider $l = 0$. In this case, 
\begin{eqnarray}
\abs{b_0} &\leq & n\sum_{m=0}^n 2^m\frac{(n+m-1)!}{(n-m)!(2m)!} \frac{\abs{E_0-E_2}^{m}}{\norm{H}_2^{m}} \notag \\
&=& T_n(z_1).
\end{eqnarray}
Then, for an arbitrary $l$, 
\begin{eqnarray}
\abs{b_l} &\leq & n^l T_n(z_1),
\end{eqnarray}
where the inequality 
\begin{eqnarray}
\frac{m!}{(m-l)!l!} &\leq & m^l \leq n^l
\end{eqnarray}
has been used. 

Finally, taking $d = n+1$ and 
\begin{eqnarray}
a_k = \frac{c_kb_{k-1}}{T_n(z_1)\norm{H}_2^{k-1}},
\end{eqnarray}
we have 
\begin{eqnarray}
\sum_{k=1}^d a_kc_k^{-1}f_k = \frac{T_n(Z)}{T_n(z_1)}e^{-\frac{1}{2}(H-E_0)^2\tau^2},
\end{eqnarray}
where $f_k$ are operators defined in Eq.~(\ref{eq:fk}). Because of the upper bound 
\begin{eqnarray}
\abs{a_k} &\leq & 2 \left(\frac{k-1}{e\tau^2}\right)^{\frac{k-1}{2}} \left(\frac{n}{\norm{H}_2}\right)^{k-1} \notag \\
&\leq & 2 \left(\frac{n^3}{e\norm{H}_2^2\tau^2}\right)^{\frac{k-1}{2}},
\end{eqnarray}
the overhead factor is 
\begin{eqnarray}\label{eq:tau}
\gamma &\approx & \left(\bfa^\dag \bfa\right)^2 \leq 4\sum_{k=1}^d \left(\frac{n^3}{e\norm{H}_2^2\tau^2}\right)^{k-1} \notag \\
&\leq & 4\frac{1}{1-n^3/(e\norm{H}_2^2\tau^2)}.
\end{eqnarray}

\section{The choice of $\tau$\label{app:tau}}

There are two bounds determining the range of $\tau$. First, we take $\tau>\sqrt{\frac{d-1}{e}}$, such that the upper bound of $\norm{f_k}_2$ decrease exponentially with $k$. Second, under the assumption $\abs{E_0-E_g} \leq \epsilon_0$, we take $\tau \leq \frac{\sqrt{d-1}}{\epsilon_0}$. To justify the second bound, we plot the absolute value of the rescaled Gaussian-power function $\abs{f'_k(x)}=\abs{f_k(x)}/c_k$ with different $k$ in Fig.~\ref{fig:Gaussian-power}, where the instance $(\mathrm{Heisenberg},\mathrm{chain},d=5)$ is taken as an example. This function has maximum points at $\left(\pm\frac{\sqrt{k-1}}{\tau},(\frac{k-1}{e\tau^2})^{\frac{k-1}{2}}/c_k\right)$. When $k$ increases from $1$ to $d$, the peaks of $f'_k(x)$ span in the range $\left[-\frac{\sqrt{d-1}}{\tau},\frac{\sqrt{d-1}}{\tau}\right]$, making the basis a filter~\cite{wall1995}. It is preferred that the ground-state energy is in the range, which leads to the second bound. Notice that we can always normalise the Hamiltonian by taking $H\leftarrow H/h_{tot}$. Then, $\norm{H}_2\leq 1$, $\abs{E_0-E_g} \leq 1$ for all $E_0\in [-1,0]$, and the two bounds are consistent. 

\begin{figure}[htbp]
\centering
\includegraphics[width=\linewidth]{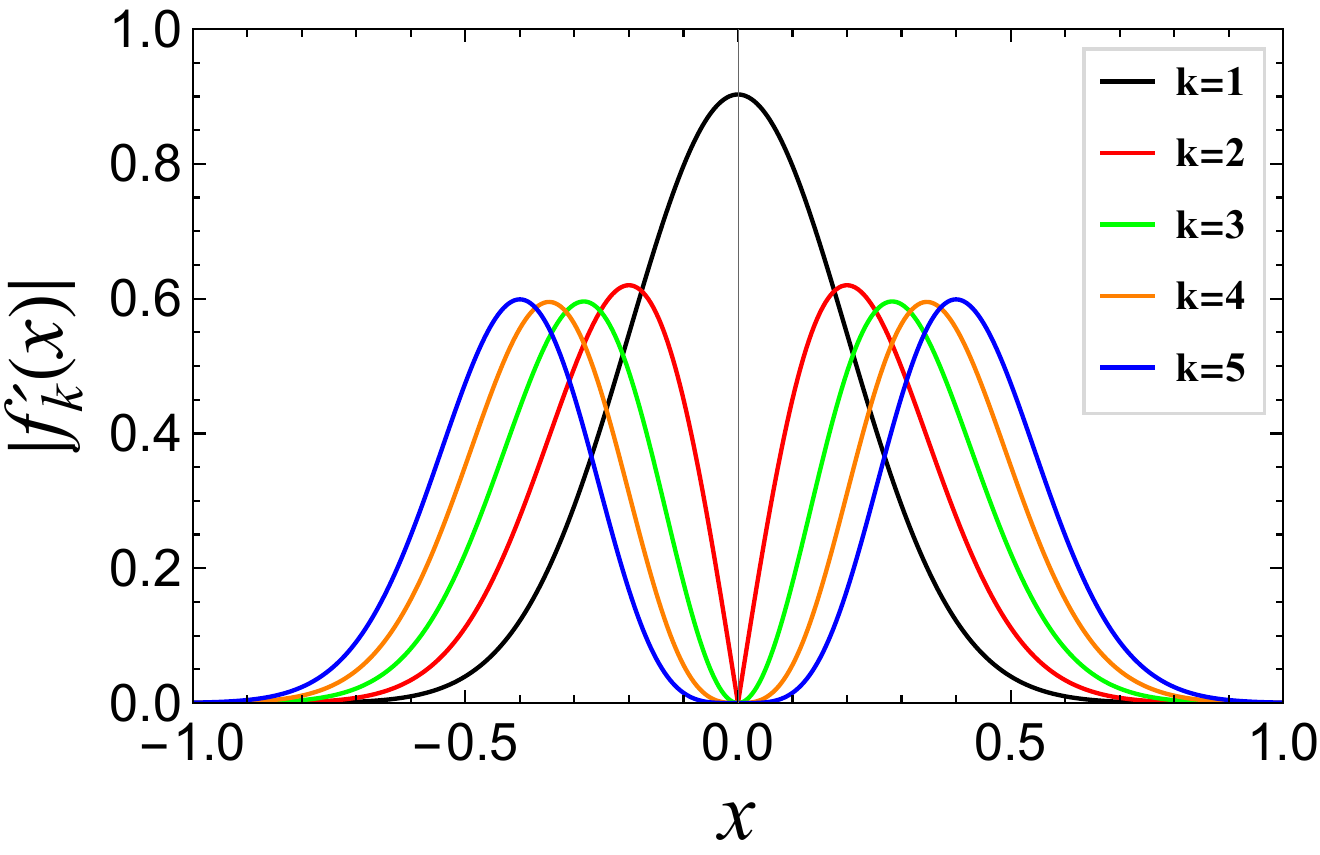}
\caption{The absolute value of the rescaled Gaussian-power function $\abs{f'_k(x)}=\abs{x^{k-1}e^{-x^2\tau^2/2}}/c_k$ with $k=1$ to $d$ and $\tau=5$. Here, we take the instance $(\mathrm{Heisenberg},\mathrm{chain},d=5)$ as an example.}
\label{fig:Gaussian-power}
\end{figure}

\section{Numerical results of the measurement number}
\label{app:practice}

In Theorem \ref{the}, we give a theoretical result about the sufficient measurement number, which is used in later numerical study. In this section, we numerically calculate the necessary measurement number and compare the regularisation method to the thresholding method. 

\subsection{Necessary measurement costs}

Due to the statistical error, the output energy $\hat{E}_{min}$ follows a certain distribution in the vicinity of the true ground-state energy $E_g$. Given the permissible energy error $\epsilon$ and failure probability $\kappa$, a measurement number is said to be sufficient if $\hat{E}_{min}$ is in the range $[E_g-\epsilon,E_g+\epsilon]$ with a probability not lower than $1-\kappa$; and the necessary measurement number is the minimum sufficient measurement number. 

Taking the instance $(\mathrm{Heisenberg},\mathrm{chain},d=5)$ as an example, Fig.~\ref{fig:M-epsilon} shows the errors in the ground-state energy and the corresponding necessary/sufficient measurement numbers. We can find that the sufficient measurement costs are close to the necessary measurement numbers. 

\begin{figure}[htbp]
\centering
\includegraphics[width=\linewidth]{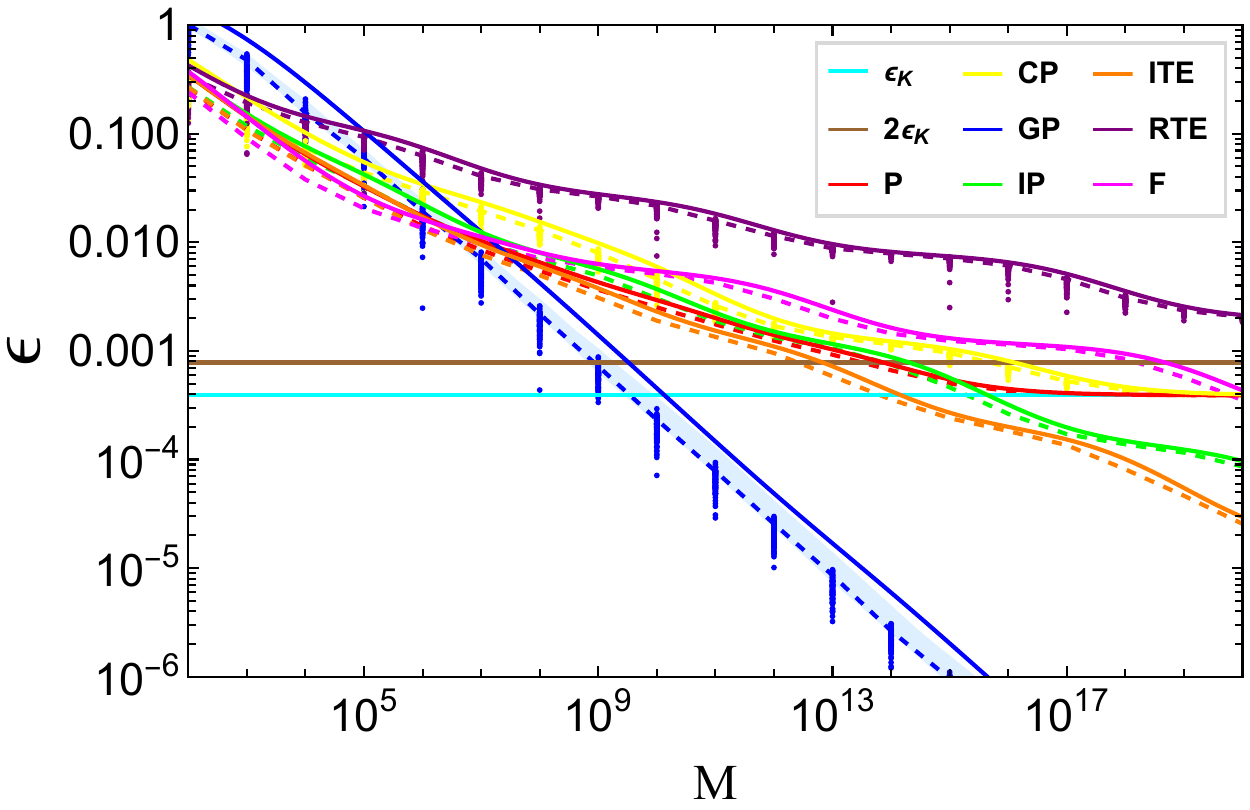}
\caption{The variational ground-state energy errors $\epsilon$ and the corresponding measurement numbers $M$ for the quantum KSD algorithms listed in Table~\ref{table}. The instance is $(\mathrm{Heisenberg},\mathrm{chain},d=5)$ and $E_g=-1$. The solid curves indicate the sufficient measurement numbers computed through Theorem~\ref{the}. The dots denote errors at various measurement numbers computed following Algorithm~\ref{alg:QKSD}. At each measurement number $M$, Algorithm~\ref{alg:QKSD} is implemented for $100$ times (simulated on a classical computer), and $(M,\abs{\hat{E}_{min}-E_g})$ of each implementation is plotted in the figure. The dashed lines are taken such that $10\%$ dots are above them, i.e. the dashed lines corresponding to the necessary measurement number for the failure probability $\kappa=0.1$. For the GP algorithm, the solid and dashed lines represent the result of the GP algorithm with $E_0 = E_g$. The light blue area illustrates the range of necessary measurement cost when we take $E_0\in [E_g-0.1,E_g+0.1]$ in the GP algorithm.}
\label{fig:M-epsilon}
\end{figure}

When estimating the necessary measurement numbers, we need to numerically simulate the estimators $\hat{\bfH}$ and $\hat{\bfS}$, which can be regarded as the summations of the matrices $\bfH$ and $\bfS$ and random noise matrices $\hat{\bfH}-\bfH$ and $\hat{\bfS}-\bfS$, respectively. Considering Gaussian statistical error and the collective measurement protocol, the random matrices $\hat{\bfH}-\bfH$ and $\hat{\bfS}-\bfS$ can be numerically generated. The real and imaginary parts of the random matrix entries obey the Gaussian distribution with the mean value $0$ (assume the estimation is unbiased for all algorithms) and the standard deviations are $C_{\bfH}/\sqrt{M}$ and $C_{\bfS}/\sqrt{M}$, respectively. Under collective measurement protocols, the random matrices have the same structure as $\bfH$ and $\bfS$. 

\subsection{Comparison using the thresholding method}

Ref.~\cite{epperly2022} provides a detailed analysis of the thresholding procedure. The optimal thresholds are roughly close to the noise rate. Here, we set them to $10C_{\bfS}/\sqrt{M}$ to get relatively stable results. With the same instance $(\mathrm{Heisenberg},\mathrm{chain},d=5)$, the results shown in Fig.~\ref{fig:thresholding} are similar to Fig.~\ref{fig:M-epsilon}. This indicates that regardless of which method is used to deal with the ill-conditioned problem, GP algorithm outperforms other algorithms in terms of measurement costs. 

\begin{figure}[htbp]
\centering
\includegraphics[width=\linewidth]{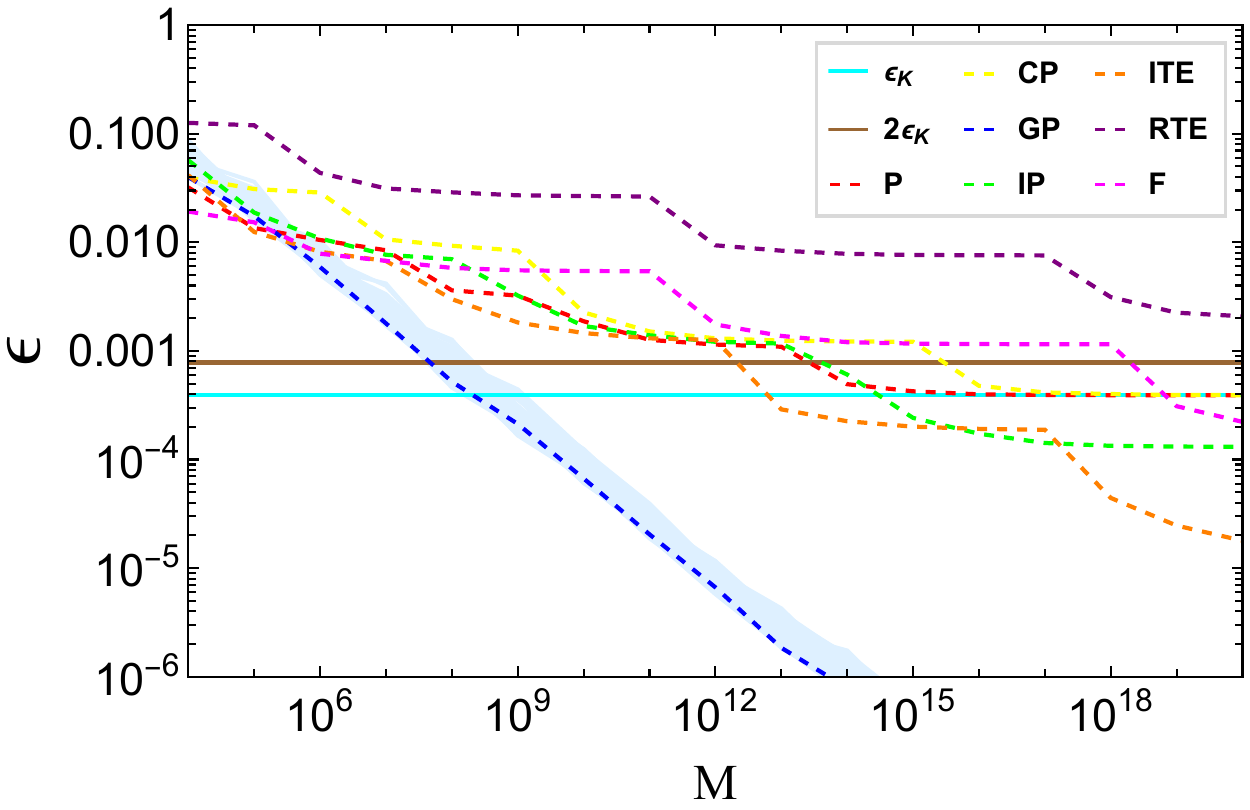}
\caption{The absolute energy error $\epsilon$ and the corresponding necessary measurement numbers $M$ when using the thresholding method. The instance is $(\mathrm{Heisenberg},\mathrm{chain},d=5)$ and $E_g=-1$.
Similar to Fig.~\ref{fig:M-epsilon}, the dashed lines indicate the necessary measurement costs with $\kappa=0.1$. The blue dashed line represents the result of the GP algorithm with $E_0 = E_g$. The light blue area illustrates the range of necessary measurement cost when we take $E_0\in [E_g-0.1,E_g+0.1]$ in the GP algorithm.}
\label{fig:thresholding}
\end{figure}

\begin{figure*}[htbp]
\centering
\includegraphics[width=0.85\linewidth]{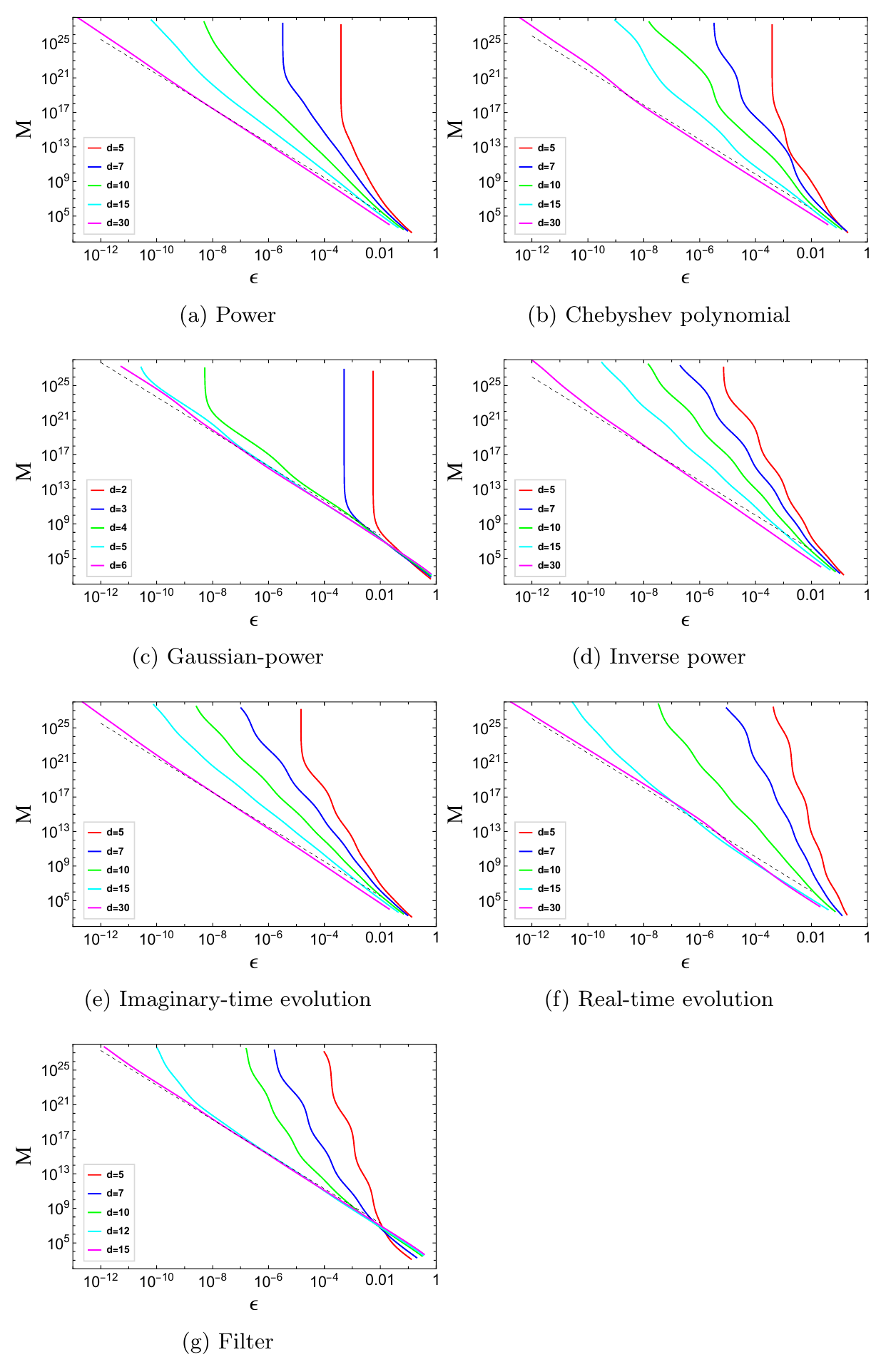}
\caption{The estimated measurement number $M$ and the corresponding energy error $\epsilon$ at different $d$ for quantum Krylov subspace diagonalisation algorithms. We take the $10$-qubit Heisenberg chain model as an example. The dashed lines indicate the $1/\sqrt{M}$ scaling. The failure rate $\kappa=0.1$. For the GP basis, $E_0$ is randomly chosen in $[E_g-0.1, E_g+0.1]$. For other bases, we use optimal parameters and minimum cost factors.}
\label{fig:scaling}
\end{figure*}

\bibliographystyle{quantum}
\bibliography{reference}

\begin{thebibliography}{10}

\bibitem{dagotto1994}
Elbio Dagotto.
\newblock ``Correlated electrons in high-temperature superconductors''.
\newblock \href{https://dx.doi.org/10.1103/RevModPhys.66.763}{Reviews of Modern Physics {\bf 66}, 763}~(1994).

\bibitem{wall1995}
Michael~R Wall and Daniel Neuhauser.
\newblock ``Extraction, through filter-diagonalization, of general quantum eigenvalues or classical normal mode frequencies from a small number of residues or a short-time segment of a signal. i. theory and application to a quantum-dynamics model''.
\newblock \href{https://dx.doi.org/10.1063/1.468999}{The Journal of chemical physics {\bf 102}, 8011--8022}~(1995).

\bibitem{caurier2005}
Etienne Caurier, Gabriel Mart{\'\i}nez-Pinedo, Fr{\'e}deric Nowacki, Alfredo Poves, and AP~Zuker.
\newblock ``The shell model as a unified view of nuclear structure''.
\newblock \href{https://dx.doi.org/10.1103/RevModPhys.77.427}{Reviews of modern Physics {\bf 77}, 427}~(2005).

\bibitem{lanczos1950}
Cornelius Lanczos.
\newblock ``An iteration method for the solution of the eigenvalue problem of linear differential and integral operators''.
\newblock \href{https://dx.doi.org/10.6028/jres.045.026}{Journal of research of the National Bureau of Standards {\bf 45}, 255--282}~(1950).

\bibitem{saad1992}
Yousef Saad.
\newblock ``Numerical methods for large eigenvalue problems''.
\newblock \href{https://dx.doi.org/10.1137/1.9781611970739}{Society for Industrial and Applied Mathematics}. ~(2011).

\bibitem{bjorck2015}
{\AA}ke Bj{\"o}rck.
\newblock ``Numerical methods in matrix computations''.
\newblock \href{https://dx.doi.org/10.1007/978-3-319-05089-8}{Volume~59 of Texts in Applied Mathematics}.
\newblock Springer. ~(2015).

\bibitem{avella2012}
Adolfo Avella and Ferdinando Mancini.
\newblock ``Strongly correlated systems''.
\newblock \href{https://dx.doi.org/10.1007/978-3-642-35106-8}{Volume 176 of Springer Series in Solid-State Sciences}.
\newblock Springer. ~(2013).

\bibitem{motta2019}
Mario Motta, Chong Sun, Adrian T.~K. Tan, Matthew~J. O'Rourke, Erika Ye, Austin~J. Minnich, Fernando G. S.~L. Brand{\~{a}}o, and Garnet Kin-Lic Chan.
\newblock ``Determining eigenstates and thermal states on a quantum computer using quantum imaginary time evolution''.
\newblock \href{https://dx.doi.org/10.1038/s41567-019-0704-4}{Nature Physics {\bf 16}, 205--210}~(2019).

\bibitem{yeter2020}
K{\"u}bra Yeter-Aydeniz, Raphael~C Pooser, and George Siopsis.
\newblock ``Practical quantum computation of chemical and nuclear energy levels using quantum imaginary time evolution and lanczos algorithms''.
\newblock \href{https://dx.doi.org/10.1038/s41534-020-00290-1}{npj Quantum Information {\bf 6}, 1--8}~(2020).

\bibitem{parrish2019}
Robert~M. Parrish and Peter~L. McMahon.
\newblock ``Quantum filter diagonalization: Quantum eigendecomposition without full quantum phase estimation''~(2019).
\newblock  \href{http://arxiv.org/abs/1909.08925}{arXiv:1909.08925}.

\bibitem{stair2020}
Nicholas~H. Stair, Renke Huang, and Francesco~A. Evangelista.
\newblock ``A multireference quantum krylov algorithm for strongly correlated electrons''.
\newblock \href{https://dx.doi.org/10.1021/acs.jctc.9b01125}{Journal of Chemical Theory and Computation {\bf 16}, 2236--2245}~(2020).

\bibitem{bespalova2021}
Tatiana~A Bespalova and Oleksandr Kyriienko.
\newblock ``Hamiltonian operator approximation for energy measurement and ground-state preparation''.
\newblock \href{https://dx.doi.org/10.1103/PRXQuantum.2.030318}{PRX Quantum {\bf 2}, 030318}~(2021).

\bibitem{cohn2021}
Jeffrey Cohn, Mario Motta, and Robert~M Parrish.
\newblock ``Quantum filter diagonalization with compressed double-factorized hamiltonians''.
\newblock \href{https://dx.doi.org/10.1103/PRXQuantum.2.040352}{PRX Quantum {\bf 2}, 040352}~(2021).

\bibitem{klymko2022}
Katherine Klymko, Carlos Mejuto-Zaera, Stephen~J Cotton, Filip Wudarski, Miroslav Urbanek, Diptarka Hait, Martin Head-Gordon, K~Birgitta Whaley, Jonathan Moussa, Nathan Wiebe, et~al.
\newblock ``Real-time evolution for ultracompact hamiltonian eigenstates on quantum hardware''.
\newblock \href{https://dx.doi.org/10.1103/PRXQuantum.3.020323}{PRX Quantum {\bf 3}, 020323}~(2022).

\bibitem{cortes2022}
Cristian~L Cortes and Stephen~K Gray.
\newblock ``Quantum krylov subspace algorithms for ground-and excited-state energy estimation''.
\newblock \href{https://dx.doi.org/10.1103/PhysRevA.105.022417}{Physical Review A {\bf 105}, 022417}~(2022).

\bibitem{epperly2022}
Ethan~N Epperly, Lin Lin, and Yuji Nakatsukasa.
\newblock ``A theory of quantum subspace diagonalization''.
\newblock \href{https://dx.doi.org/10.1137/21M145954X}{SIAM Journal on Matrix Analysis and Applications {\bf 43}, 1263--1290}~(2022).

\bibitem{shen2022}
Yizhi Shen, Katherine Klymko, James Sud, David~B. Williams-Young, Wibe A.~de Jong, and Norm~M. Tubman.
\newblock ``Real-{T}ime {K}rylov {T}heory for {Q}uantum {C}omputing {A}lgorithms''.
\newblock \href{https://dx.doi.org/10.22331/q-2023-07-25-1066}{{Quantum} {\bf 7}, 1066}~(2023).

\bibitem{kyriienko2020}
Oleksandr Kyriienko.
\newblock ``Quantum inverse iteration algorithm for programmable quantum simulators''.
\newblock \href{https://dx.doi.org/10.1038/s41534-019-0239-7}{npj Quantum Information {\bf 6}, 1--8}~(2020).

\bibitem{seki2021}
Kazuhiro Seki and Seiji Yunoki.
\newblock ``Quantum power method by a superposition of time-evolved states''.
\newblock \href{https://dx.doi.org/10.1103/PRXQuantum.2.010333}{PRX Quantum {\bf 2}, 010333}~(2021).

\bibitem{kirby2022}
William Kirby, Mario Motta, and Antonio Mezzacapo.
\newblock ``Exact and efficient {L}anczos method on a quantum computer''.
\newblock \href{https://dx.doi.org/10.22331/q-2023-05-23-1018}{{Quantum} {\bf 7}, 1018}~(2023).

\bibitem{babbush2018}
Ryan Babbush, Craig Gidney, Dominic~W. Berry, Nathan Wiebe, Jarrod McClean, Alexandru Paler, Austin Fowler, and Hartmut Neven.
\newblock ``Encoding electronic spectra in quantum circuits with linear t complexity''.
\newblock \href{https://dx.doi.org/10.1103/PhysRevX.8.041015}{Phys. Rev. X {\bf 8}, 041015}~(2018).

\bibitem{lee2021}
Joonho Lee, Dominic~W. Berry, Craig Gidney, William~J. Huggins, Jarrod~R. McClean, Nathan Wiebe, and Ryan Babbush.
\newblock ``Even more efficient quantum computations of chemistry through tensor hypercontraction''.
\newblock \href{https://dx.doi.org/10.1103/PRXQuantum.2.030305}{PRX Quantum {\bf 2}, 030305}~(2021).

\bibitem{peruzzo2014}
Alberto Peruzzo, Jarrod McClean, Peter Shadbolt, Man-Hong Yung, Xiao-Qi Zhou, Peter~J Love, Al{\'a}n Aspuru-Guzik, and Jeremy~L O’brien.
\newblock ``A variational eigenvalue solver on a photonic quantum processor''.
\newblock \href{https://dx.doi.org/10.1038/ncomms5213}{Nature communications {\bf 5}, 4213}~(2014).

\bibitem{mcclean2018}
Jarrod~R McClean, Sergio Boixo, Vadim~N Smelyanskiy, Ryan Babbush, and Hartmut Neven.
\newblock ``Barren plateaus in quantum neural network training landscapes''.
\newblock \href{https://dx.doi.org/10.1038/s41467-018-07090-4}{Nature communications {\bf 9}, 4812}~(2018).

\bibitem{parlett1998}
Beresford~N. Parlett.
\newblock ``The symmetric eigenvalue problem''.
\newblock \href{https://dx.doi.org/10.1137/1.9781611971163}{Society for Industrial and Applied Mathematics}. ~(1998).

\bibitem{parrish2019a}
Robert~M Parrish, Edward~G Hohenstein, Peter~L McMahon, and Todd~J Mart{\'\i}nez.
\newblock ``Quantum computation of electronic transitions using a variational quantum eigensolver''.
\newblock \href{https://dx.doi.org/10.1103/PhysRevLett.122.230401}{Physical review letters {\bf 122}, 230401}~(2019).

\bibitem{nakanishi2019}
Ken~M Nakanishi, Kosuke Mitarai, and Keisuke Fujii.
\newblock ``Subspace-search variational quantum eigensolver for excited states''.
\newblock \href{https://dx.doi.org/10.1103/PhysRevResearch.1.033062}{Physical Review Research {\bf 1}, 033062}~(2019).

\bibitem{huggins2020}
William~J Huggins, Joonho Lee, Unpil Baek, Bryan O’Gorman, and K~Birgitta Whaley.
\newblock ``A non-orthogonal variational quantum eigensolver''.
\newblock \href{https://dx.doi.org/10.1088/1367-2630/ab867b}{New Journal of Physics {\bf 22}, 073009}~(2020).

\bibitem{stair2022}
Nicholas~H. Stair, Cristian~L. Cortes, Robert~M. Parrish, Jeffrey Cohn, and Mario Motta.
\newblock ``Stochastic quantum krylov protocol with double-factorized hamiltonians''.
\newblock \href{https://dx.doi.org/10.1103/PhysRevA.107.032414}{Phys. Rev. A {\bf 107}, 032414}~(2023).

\bibitem{mcclean2017}
Jarrod~R. McClean, Mollie~E. Kimchi-Schwartz, Jonathan Carter, and Wibe~A. de~Jong.
\newblock ``Hybrid quantum-classical hierarchy for mitigation of decoherence and determination of excited states''.
\newblock \href{https://dx.doi.org/10.1103/PhysRevA.95.042308}{Phys. Rev. A {\bf 95}, 042308}~(2017).

\bibitem{mcclean2020}
Jarrod~R McClean, Zhang Jiang, Nicholas~C Rubin, Ryan Babbush, and Hartmut Neven.
\newblock ``Decoding quantum errors with subspace expansions''.
\newblock \href{https://dx.doi.org/10.1038/s41467-020-14341-w}{Nature communications {\bf 11}, 636}~(2020).

\bibitem{yoshioka2022}
Nobuyuki Yoshioka, Hideaki Hakoshima, Yuichiro Matsuzaki, Yuuki Tokunaga, Yasunari Suzuki, and Suguru Endo.
\newblock ``Generalized quantum subspace expansion''.
\newblock \href{https://dx.doi.org/10.1103/PhysRevLett.129.020502}{Phys. Rev. Lett. {\bf 129}, 020502}~(2022).

\bibitem{ekert2002}
Artur~K Ekert, Carolina~Moura Alves, Daniel~KL Oi, Micha{\l} Horodecki, Pawe{\l} Horodecki, and Leong~Chuan Kwek.
\newblock ``Direct estimations of linear and nonlinear functionals of a quantum state''.
\newblock \href{https://dx.doi.org/10.1103/PhysRevLett.88.217901}{Physical review letters {\bf 88}, 217901}~(2002).

\bibitem{huo2021}
Mingxia Huo and Ying Li.
\newblock ``Error-resilient monte carlo quantum simulation of imaginary time''.
\newblock \href{https://dx.doi.org/10.22331/q-2023-02-09-916}{Quantum {\bf 7}, 916}~(2023).

\bibitem{zeng2021}
Pei Zeng, Jinzhao Sun, and Xiao Yuan.
\newblock ``Universal quantum algorithmic cooling on a quantum computer''~(2021).
\newblock  \href{http://arxiv.org/abs/2109.15304}{arXiv:2109.15304}.

\bibitem{childs2012}
Andrew~M Childs and Nathan Wiebe.
\newblock ``Hamiltonian simulation using linear combinations of unitary operations''.
\newblock \href{https://dx.doi.org/10.26421/QIC12.11-12-1}{Quantum Information \& Computation {\bf 12}, 901--924}~(2012).

\bibitem{faehrmann2022}
Paul~K. Faehrmann, Mark Steudtner, Richard Kueng, Maria Kieferova, and Jens Eisert.
\newblock ``Randomizing multi-product formulas for {H}amiltonian simulation''.
\newblock \href{https://dx.doi.org/10.22331/q-2022-09-19-806}{{Quantum} {\bf 6}, 806}~(2022).

\bibitem{o2021}
Thomas~E O’Brien, Stefano Polla, Nicholas~C Rubin, William~J Huggins, Sam McArdle, Sergio Boixo, Jarrod~R McClean, and Ryan Babbush.
\newblock ``Error mitigation via verified phase estimation''.
\newblock \href{https://dx.doi.org/10.1103/PRXQuantum.2.020317}{PRX Quantum {\bf 2}, 020317}~(2021).

\bibitem{lu2021}
Sirui Lu, Mari~Carmen Ba{\~n}uls, and J~Ignacio Cirac.
\newblock ``Algorithms for quantum simulation at finite energies''.
\newblock \href{https://dx.doi.org/10.1103/PRXQuantum.2.020321}{PRX Quantum {\bf 2}, 020321}~(2021).

\bibitem{arfken2013}
G.B. Arfken, H.J. Weber, and F.E. Harris.
\newblock ``Mathematical methods for physicists: A comprehensive guide''.
\newblock \href{https://dx.doi.org/10.1016/C2009-0-30629-7}{Elsevier}. ~(2012).

\bibitem{lloyd1996}
Seth Lloyd.
\newblock ``Universal quantum simulators''.
\newblock \href{https://dx.doi.org/10.1126/science.273.5278.107}{Science {\bf 273}, 1073--1078}~(1996).

\bibitem{berry2007}
Dominic~W Berry, Graeme Ahokas, Richard Cleve, and Barry~C Sanders.
\newblock ``Efficient quantum algorithms for simulating sparse hamiltonians''.
\newblock \href{https://dx.doi.org/10.1007/s00220-006-0150-x}{Communications in Mathematical Physics {\bf 270}, 359--371}~(2007).

\bibitem{wiebe2010}
Nathan Wiebe, Dominic Berry, Peter H{\o}yer, and Barry~C Sanders.
\newblock ``Higher order decompositions of ordered operator exponentials''.
\newblock \href{https://dx.doi.org/10.1088/1751-8113/43/6/065203}{Journal of Physics A: Mathematical and Theoretical {\bf 43}, 065203}~(2010).

\bibitem{berry2015}
Dominic~W Berry, Andrew~M Childs, Richard Cleve, Robin Kothari, and Rolando~D Somma.
\newblock ``Simulating hamiltonian dynamics with a truncated taylor series''.
\newblock \href{https://dx.doi.org/10.1103/PhysRevLett.114.090502}{Physical review letters {\bf 114}, 090502}~(2015).

\bibitem{meister2022}
Richard Meister, Simon~C. Benjamin, and Earl~T. Campbell.
\newblock ``Tailoring {T}erm {T}runcations for {E}lectronic {S}tructure {C}alculations {U}sing a {L}inear {C}ombination of {U}nitaries''.
\newblock \href{https://dx.doi.org/10.22331/q-2022-02-02-637}{{Quantum} {\bf 6}, 637}~(2022).

\bibitem{campbell2019}
Earl Campbell.
\newblock ``Random compiler for fast hamiltonian simulation''.
\newblock \href{https://dx.doi.org/10.1103/PhysRevLett.123.070503}{Physical review letters {\bf 123}, 070503}~(2019).

\bibitem{yang2021}
Yongdan Yang, Bing-Nan Lu, and Ying Li.
\newblock ``Accelerated quantum monte carlo with mitigated error on noisy quantum computer''.
\newblock \href{https://dx.doi.org/10.1103/PRXQuantum.2.040361}{PRX Quantum {\bf 2}, 040361}~(2021).

\bibitem{anderson1987}
Philip~W Anderson.
\newblock ``The resonating valence bond state in la2cuo4 and superconductivity''.
\newblock \href{https://dx.doi.org/10.1126/science.235.4793.1196}{science {\bf 235}, 1196--1198}~(1987).

\bibitem{lee2006}
Patrick~A Lee, Naoto Nagaosa, and Xiao-Gang Wen.
\newblock ``Doping a mott insulator: Physics of high-temperature superconductivity''.
\newblock \href{https://dx.doi.org/10.1103/RevModPhys.78.17}{Reviews of modern physics {\bf 78}, 17}~(2006).

\bibitem{seki2020}
Kazuhiro Seki, Tomonori Shirakawa, and Seiji Yunoki.
\newblock ``Symmetry-adapted variational quantum eigensolver''.
\newblock \href{https://dx.doi.org/10.1103/PhysRevA.101.052340}{Phys. Rev. A {\bf 101}, 052340}~(2020).

\bibitem{horn2012}
Roger~A. Horn and Charles~R. Johnson.
\newblock ``Matrix analysis''.
\newblock \href{https://dx.doi.org/10.1017/CBO9781139020411}{Cambridge University Press}. ~(2012).
\newblock 2 edition.

\bibitem{bjorck2008}
Germund Dahlquist and {\AA}ke Bj{\"o}rck.
\newblock ``Numerical methods in scientific computing, volume i''.
\newblock \href{https://dx.doi.org/10.1137/1.9780898717785}{Society for Industrial and Applied Mathematics}. ~(2008).

\bibitem{kirby2024analysis}
William Kirby.
\newblock ``Analysis of quantum krylov algorithms with errors''~(2024).
\newblock  \href{http://arxiv.org/abs/2401.01246}{arXiv:2401.01246}.

\bibitem{code}
\url{https://github.com/ZongkangZhang/QKSD}.

\bibitem{hoeffding1963}
Wassily Hoeffding.
\newblock ``Probability inequalities for sums of bounded random variables''.
\newblock \href{https://dx.doi.org/10.1080/01621459.1963.10500830}{Journal of the American Statistical Association {\bf 58}, 13--30}~(1963).

\bibitem{tropp2015}
Joel~A. Tropp.
\newblock ``An introduction to matrix concentration inequalities''.
\newblock \href{https://dx.doi.org/10.1561/2200000048}{Foundations and Trends{\textregistered} in Machine Learning {\bf 8}, 1--230}~(2015).

\bibitem{raymundo2021}
Raymundo Albert and Cecilia~G. Galarza.
\newblock ``Model order selection for sum of complex exponentials''.
\newblock In 2021 IEEE URUCON.
\newblock \href{https://dx.doi.org/10.1109/URUCON53396.2021.9647257}{Pages 561--565}.
\newblock ~(2021).

\bibitem{babusci2012}
D.~Babusci, G.~Dattoli, and M.~Quattromini.
\newblock ``On integrals involving hermite polynomials''.
\newblock \href{https://dx.doi.org/https://doi.org/10.1016/j.aml.2012.02.043}{Applied Mathematics Letters {\bf 25}, 1157--1160}~(2012).

\bibitem{kitaev2002}
Alexei~Yu Kitaev, Alexander Shen, and Mikhail~N Vyalyi.
\newblock ``Classical and quantum computation''.
\newblock \href{https://dx.doi.org/10.1090/gsm/047}{Graduate studies in mathematics}. American Mathematical Society. ~(2002).

\bibitem{kempe2006}
Julia Kempe, Alexei Kitaev, and Oded Regev.
\newblock ``The complexity of the local hamiltonian problem''.
\newblock \href{https://dx.doi.org/10.1137/S0097539704445226}{Siam journal on computing {\bf 35}, 1070--1097}~(2006).

\bibitem{gharibian2022}
Chris Cade, Marten Folkertsma, Sevag Gharibian, Ryu Hayakawa, Fran\c{c}ois Le~Gall, Tomoyuki Morimae, and Jordi Weggemans.
\newblock ``{Improved Hardness Results for the Guided Local Hamiltonian Problem}''.
\newblock In Kousha Etessami, Uriel Feige, and Gabriele Puppis, editors, 50th International Colloquium on Automata, Languages, and Programming (ICALP 2023).
\newblock \href{https://dx.doi.org/10.4230/LIPIcs.ICALP.2023.32}{Volume 261 of Leibniz International Proceedings in Informatics (LIPIcs), pages 32:1--32:19}.
\newblock Dagstuhl, Germany~(2023). Schloss Dagstuhl -- Leibniz-Zentrum f{\"u}r Informatik.

\bibitem{lee2022}
Seunghoon Lee, Joonho Lee, Huanchen Zhai, Yu~Tong, Alexander~M Dalzell, Ashutosh Kumar, Phillip Helms, Johnnie Gray, Zhi-Hao Cui, Wenyuan Liu, et~al.
\newblock ``Evaluating the evidence for exponential quantum advantage in ground-state quantum chemistry''.
\newblock \href{https://dx.doi.org/10.1038/s41467-023-37587-6}{Nature communications {\bf 14}, 1952}~(2023).

\end{thebibliography}

\end{document}